\documentclass[sigconf]{acmart}

\usepackage{amsmath,amsfonts}
\usepackage{mathrsfs}
\usepackage{bm}
\usepackage{nicefrac}
\usepackage{booktabs}
\usepackage{array}
\usepackage{multirow}
\usepackage{threeparttable}
\usepackage{makecell}
\usepackage[procnumbered,ruled,vlined,linesnumbered]{algorithm2e}
\usepackage{siunitx}
\usepackage{stfloats}
\usepackage{subfigure}
\usepackage{array}
\usepackage{balance}
\usepackage{tabularx}
\usepackage{lipsum}
\usepackage{xcolor}
\usepackage{tikz,tikz-cd}
\usepackage{pgfplots,capt-of}

\pgfplotsset{compat=1.18}

\newcommand{\spar}[1]{\left( #1 \right)}
\newcommand{\normspar}[1]{( #1 )}
\newcommand{\midpar}[1]{\left[ #1 \right]}

\newcommand{\myord}[1]{{#1}^{\rm{th}}}
\newcommand{\setof}[1]{\left\{#1 \right\}}
\newcommand{\abs}[1]{\left|#1 \right|}
\newcommand{\mean}[1]{\mathbb{E}\midpar{#1}}
\newcommand{\ceil}[1]{\left\lceil#1\right\rceil}
\newcommand{\floor}[1]{\left\lfloor#1\right\rfloor}
\newcommand{\Abs}[1]{\left\Vert#1\right\Vert}
\newcommand{\trace}[1]{\mathrm{Tr}\spar{#1}}
\newcommand{\diag}[1]{\mathrm{diag}\spar{#1}}

\newcommand{\hit}[2]{H_{{#1}{#2}}}
\newcommand{\absorb}[1]{H_{#1}}

\newcommand{\randhit}[2]{T_{{#1}{#2}}}

\newcommand{\algoname}[1]{\textsc{#1}}
\newcommand{\bernfunc}[1]{f\spar{#1}}
\newcommand{\bigo}[1]{O\spar{#1}}

\newcommand{\bigspar}[1]{\big( #1 \big)}
\newcommand{\biggspar}[1]{\bigg( #1 \bigg)}

\newcommand{\bigabs}[1]{\big| #1 \big|}
\newcommand{\biggabs}[1]{\bigg| #1 \bigg|}

\newcommand{\bigtrace}[1]{\mathrm{Tr}\bigspar{#1}}
\newcommand{\normtrace}[1]{\mathrm{Tr}\normspar{#1}}

\newcommand{\bigmean}[1]{\mathbb{E}\big[ #1 \big]}
\newcommand{\bigbigo}[1]{O\bigspar{#1}}
\newcommand{\bigbernfunc}[1]{f\bigspar{#1}}

\newcommand{\gr}{G}
\newcommand{\kem}{K}

\newcommand{\sqrtset}{\mathcal{X}}
\newcommand{\rea}{\mathbb{R}}

\newcommand{\rme}{\mathrm{e}}
\newcommand{\dmax}{d_{\max}}

\newcommand{\subfund}[1][s]{\spar{\mati-\matp_{-{#1}}}^{-1}}
\newcommand{\normfund}{\tilde{\fund}}
\newcommand{\diaglap}{\tilde{\lap}}
\newcommand{\normvecpi}{\tilde{\vecpi}}

\newcommand{\trunckem}[1][l]{\kem^{\scriptscriptstyle (#1)}}
\newcommand{\truncfund}[1][l]{\fund^{\scriptscriptstyle (#1)}}
\newcommand{\approxkem}{\hat{\kem}^{\scriptscriptstyle (l)}}

\newcommand{\sqrtkem}{\tilde{\kem}^{\scriptscriptstyle (l)}}
\newcommand{\treekem}{\hat{\kem}}
\newcommand{\retnum}[2]{\hat{t}^{\scriptscriptstyle (l)}_{\scriptscriptstyle #1,#2}}
\newcommand{\meanretnum}[1]{\bar{t}^{\scriptscriptstyle (l)}_{\scriptscriptstyle #1}}
\newcommand{\passnum}[2]{\hat{t}_{\scriptscriptstyle #1,#2}}
\newcommand{\meanpassnum}[1]{\bar{t}_{\scriptscriptstyle #1}}

\newcommand{\veca}{\boldsymbol{a}}
\newcommand{\vecb}{\boldsymbol{b}}
\newcommand{\vecc}{\boldsymbol{c}}

\newcommand{\vece}{\boldsymbol{e}}

\newcommand{\vecpi}{\boldsymbol{\pi}}

\newcommand{\vecone}{\boldsymbol{1}}
\newcommand{\veczero}{\boldsymbol{0}}

\newcommand{\normlap}{{\boldsymbol{\mathcal{L}}}}
\newcommand{\fund}{\boldsymbol{F}}

\newcommand{\lap}{\boldsymbol{L}}
\newcommand{\mata}{\boldsymbol{A}}
\newcommand{\matb}{\boldsymbol{B}}
\newcommand{\matc}{\boldsymbol{C}}
\newcommand{\matd}{\boldsymbol{D}}

\newcommand{\mati}{\boldsymbol{I}}

\newcommand{\matp}{\boldsymbol{P}}

\newcommand{\matx}{\boldsymbol{X}}

\newcommand{\matpi}{\boldsymbol{\Pi}}

\newcommand{\lemref}[1]{Lemma~\ref{#1}}
\newcommand{\thmref}[1]{Theorem~\ref{#1}}

\newcommand{\algoref}[1]{Algorithm~\ref{#1}}

\newcommand{\secref}[1]{Section~\ref{#1}}
\newcommand{\tabref}[1]{Table~\ref{#1}}
\newcommand{\figref}[1]{Figure~\ref{#1}}

\newcommand{\lineref}[1]{Line~\ref{#1}}
\newcommand{\linerangeref}[2]{Lines~\ref{#1}-\ref{#2}}

\newcommand{\nicehalf}{\nicefrac{1}{2}}

\DeclareMathOperator*{\argmax}{arg\,max}

\DontPrintSemicolon
\SetKw{KwAnd}{and}
\SetFuncSty{textsc}
\SetKwInOut{Input}{Input\ \ \ \ }

\SetKwInOut{Output}{Output}

\AtBeginDocument{%
  }

\copyrightyear{2024}
\acmYear{2024}
\setcopyright{acmlicensed}
\acmConference[KDD '24]{Proceedings of the 30th ACM SIGKDD Conference on Knowledge Discovery and Data Mining}{August 25--29, 2024}{Barcelona, Spain.}
\acmISBN{979-8-4007-0490-1/24/08}
\acmDOI{10.1145/3637528.3671859}

\newboolean{fullver}

\begin{document}

\title{Fast Computation of Kemeny's Constant for Directed Graphs}

\author{Haisong Xia}
\affiliation{
      \institution{Fudan University}
      \city{Shanghai}
      \country{China}
}
\email{hsxia22@m.fudan.edu.cn}

\author{Zhongzhi Zhang\footnotemark}
\affiliation{
      \institution{Fudan University}
      \city{Shanghai}
      \country{China}
}
\email{zhangzz@fudan.edu.cn}

\begin{abstract}
      Kemeny's constant for random walks on a graph is defined as the mean hitting time from one node to another selected randomly according to the stationary distribution. It has found numerous applications and attracted considerable research interest. However, exact computation of Kemeny's constant requires matrix inversion, which scales poorly for large networks with millions of nodes. Existing approximation algorithms either leverage properties exclusive to undirected graphs or involve inefficient simulation, leaving room for further optimization. To address these limitations for directed graphs, we propose two novel approximation algorithms for estimating Kemeny's constant on directed graphs with theoretical error guarantees. Extensive numerical experiments on real-world networks validate the superiority of our algorithms over baseline methods in terms of efficiency and accuracy.
\end{abstract}




\maketitle

\section{Introduction}

Random walks on complex networks have emerged as a powerful analytical tool with broad applications including recommendation systems~\cite{PaBe21}, representation learning~\cite{ZhRo21}, privacy amplification~\cite{LiTaTaKaCaYo22}, and so on. For a random walk on a graph, a fundamental quantity is the hitting time \(\hit{i}{j}\), which is defined as the expected number of steps for a walker starting from node \(i\) to visit node \(j\) for the first time. As a key quantity, hitting times have been widely utilized across domains to address problems in complex networks, such as assessing transmission costs in communication networks~\cite{GaMaPrSh06}, developing clustering algorithms~\cite{ChLiTa08,Ab18}, and identifying significant nodes~\cite{WhSm03}.

Stemming from the hitting time, many important quantities for random walks can be formulated, such as Kemeny's constant. For random walks on a graph, Kemeny's constant is defined as the expected hitting time from an arbitrary source node to the target selected randomly according to the stationary distribution of the random walk. Kemeny's constant has found various applications in diverse areas. First, it is one of the widely used connectivity~\cite{BeHe19} or criticality~\cite{LeSa18} measures for a graph. Second, based on Kemeny's constant, an edge centrality~\cite{LiLiZhCaSh21,AlBiCuMePo23} was defined to identify important edges. Finally, Kemeny's constant was applied to characterize the performance of consensus protocols with noise~\cite{JaOl19}.

\renewcommand{\thefootnote}{*}
\footnotetext[1]{Corresponding author. Both authors are with Shanghai Key Laboratory of Intelligent Information Processing, School of Computer Science, Fudan University, Shanghai, 200433, China. Zhongzhi Zhang is also with Institute of Intelligent Complex Systems, Fudan University, Shanghai, 200433, China.}


Despite the utility of Kemeny's constant across various applications, directly computing it on large real-world networks remains prohibitively expensive. As discussed in \secref{subsec:hit-kem}, calculating Kemeny's constant involves matrix inversion, whose complexity is \(O(n^3)\) for an \(n\)-node graph. This cubic scaling renders exact computation infeasible for networks with millions of nodes.

In order to reduce computational time for Kemeny's constant, some approximation algorithms have been developed to estimate this graph invariant. Xu~\textit{et al.}~\cite{XuShZhKaZh20} proposed \algoname{ApproxKemeny}, which is based on Hutchinson's method and the nearly linear-time Laplacian solver~\cite{KySa16}. However, results in \secref{sec:experim} indicates that \algoname{ApproxKemeny} requires much more memory space than other methods, reducing its scalability for large networks. Very recently, Li~\textit{et al.}~\cite{LiHuLe21} provided \algoname{DynamicMC}, which is based on simulating truncated random walks. While its GPU implementation achieves state-of-the-art performance, \algoname{DynamicMC} still has many opportunities for further optimization.

On the other hand, most of existing methods for estimating Kemeny's constant  are restricted to undirected networks, including \algoname{ApproxKemeny} and \algoname{DynamicMC}. For example, the Laplacian solver~\cite{KySa16} leverages some specific properties of undirected graphs, thus \algoname{ApproxKemeny} fails to support digraphs. Although \algoname{DynamicMC} can handle digraphs,  its theoretical guarantees are not readily extended to directed graphs from the perspective in~\cite{LiHuLe21}. Nonetheless, many real-world networks are inherently directed, such as citation networks, the World Wide Web, and online social networks. The lack of an efficient approximation algorithm for estimating Kemeny's constant on directed graphs limits further applications on these important networks.

Motivated by \algoname{DynamicMC}, we provide an approximation algorithm \algoname{ImprovedMC} for estimating Kemeny's constant of digraphs with error guarantee. Apart from simulating truncated random walks, \algoname{ImprovedMC} also incorporates several optimization techniques. First, \algoname{ImprovedMC} adaptively determines the amount of simulation initialized from each node, reducing unnecessary simulation without sacrificing theoretical accuracy. Additionally, \algoname{ImprovedMC} restricts simulation to a subset of nodes in the network. By sampling from selected starting nodes, \algoname{ImprovedMC} achieves sublinear time complexity while still preserving theoretical accuracy. Extensive experiments reveal that compared with the state-of-the-art method \algoname{DynamicMC}, \algoname{ImprovedMC} attains up to \(800\times\) speedup, while achieving comparable accuracy.
To further improve the accuracy, we derive an alternative formula that connects Kemeny's constant with the inverse of a submatrix associated with the transition matrix. This motivates the development of a new Monte Carlo algorithm \algoname{TreeMC} based on directed tree sampling. We experimentally demonstrate the superiority of \algoname{TreeMC} over the state-of-the-art method in terms of both efficiency and accuracy.

The key contributions of our work are summarized as follows.

\begin{itemize}
      \item First, we develop an improved Monte Carlo algorithm \algoname{ImprovedMC} to estimate Kemeny's constant of digraphs by simulating truncated random walks. \algoname{ImprovedMC} achieves sublinear time complexity while still ensuring provable accuracy guarantees.
      \item Based on a derived alternative formula, we propose another Monte Carlo algorithm \algoname{TreeMC} that approximates Kemeny's constant of digraphs by sampling directed rooted spanning trees, which is considerably accurate.
      \item We conduct extensive experiments on real-world networks. The results indicate that both of our proposed algorithms outperform the baseline approaches by orders of magnitude speed-up, while still retaining comparable accuracy.
\end{itemize}

\section{Preliminaries}\label{sec:prelim}

\ifthenelse{\boolean{fullver}}{
      In this section, we provide an overview of key notations and essential graph concepts, including random walk, hitting time, and Kemeny's constant.
}

\subsection{Notations}

Let \(\rea\) denote the set of real numbers. We use regular lowercase letters like \(a,b,c\) for scalars within \(\rea\). Bold lowercase letters, such as \(\veca, \vecb, \vecc\), represent vectors, while bold uppercase letters, like \(\mata, \matb, \matc\), denote matrices. Specific elements are accessed by using subscripts: \(a_i\) for the \(\myord{i}\) element of \(\veca\) and \(\mata_{i,j}\) for the entry at position \((i,j)\). Subvectors and submatrices are similarly indicated with subscript numerals.
For example, \(\veca_{-i}\) denotes the subvector of \(\veca\) obtained by excluding its \(\myord{i}\) element, while \(\mata_{-i}\) represents the submatrix of \(\mata\) constructed by removing its \(\myord{i}\) row and \(\myord{i}\) column. Crucially, subscripts take precedence over superscripts in this notation. Consequently, \(\mata_{-i}^{-1}\) represents the inverse of \(\mata_{-i}\), rather than a submatrix of \(\mata^{-1}\). We use \(\vecone\) to denote a vector of specific dimensions with all elements being \(1\). \tabref{tab:notations} lists the frequently used notations throughout this paper.

\begin{table}[ht]
      \centering
      \caption{Frequently used notations.}
      \label{tab:notations}
      \begin{tabularx}{\linewidth}{rX}
            \toprule
            \textbf{Notation}                  & \textbf{Description}                                                                                       \\
            \midrule
            \(\gr=(V,E)\)                      & A digraph with node set \(V\) and edge set \(E\).                                                          \\
            \(n,m\)                            & The number of nodes and edges in \(\gr\).                                                                  \\
            \(\pi_i\)                          & The stationary distribution of node \(i\).                                                                 \\
            \(\matp\)                          & The transition matrix of random walks on \(\gr\).                                                          \\
            \(\lambda_i\)                      & The \(\myord{i}\) largest eigenvalue of \(\matp\) sorted by modulus.                                       \\
            \(\lambda\)                        & Denoted as \(\abs{\lambda_2}\).                                                                            \\
            \(\retnum{i}{j}\)                  & The returning times to node \(i\) of the \(\myord{j}\) \(l\)-truncated random walk that starts from \(i\). \\
            \(\passnum{i}{j}\)                 & The returning times to node \(i\) of the \(\myord{j}\) absorbing random walk that starts from \(i\).       \\
            \(\meanretnum{i},\meanpassnum{i}\) & The empirical mean of \(\retnum{i}{j}\) and \(\passnum{i}{j}\).                                            \\
            \bottomrule
      \end{tabularx}
\end{table}

\subsection{Graph and Random Walk}\label{subsec:graph-rand}

Let \(\gr=(V,E)\) be a digraph, where \(V\) is the set of nodes, and \(E\) is the set of edges. The digraph \(\gr\) has a total of \(n=\abs{V}\) nodes and \(m=\abs{E}\) edges. Throughout this paper, all the digraphs mentioned are assumed to be strongly connected without explicit qualification.

The adjacency matrix \(\mata\) of graph \(\gr\) mathematically encodes its topological properties. Here, the entry \(\mata_{i,j}\) represents the adjacency relation between nodes \(i\) and \(j\). \(\mata_{i,j}=1\) if there exists a directed edge pointing from \(i\) to \(j\). Conversely, the absence of such an edge is indicated by \(\mata_{i,j}=0\). In a digraph \(\gr\), the out-degree \(d_i\) of node \(i\) is defined as the number of its out-neighbours. If we denote the diagonal out-degree matrix of digraph \(\gr\) as \(\matd=\diag{d_1,d_2,\dots,d_n}\), the Laplacian matrix of \(\gr\) is defined as \(\lap=\matd-\mata\).


For a digraph \(\gr\) with \(n\) nodes, a random walk on \(\gr\) is defined through its transition matrix \(\matp\in\rea^{n\times n}\). At each step, if the walker is at node \(i\), it moves to an out-neighbour \(j\) with equal probability \(\matp_{i,j}\). It follows readily that \(\matp=\matd^{-1}\mata\).
Assuming \(\matp\) is finite, aperiodic, and irreducible, the random walk has an unique stationary distribution \(\vecpi=\left(\pi_1,\pi_2,\cdots,\pi_n\right)^\top\), satisfying \(\vecpi^\top\vecone=1\) and \(\vecpi^\top\matp=\matp\). Clearly, \(\vecpi\) is the left \(1\)-eigenvector of \(\matp\). Let \(\lambda_1,\lambda_2,\dots,\lambda_n\) be the eigenvalues of \(\matp\), where \(1=\abs{\lambda_1}>\abs{\lambda_2}\geq\cdots\geq\abs{\lambda_n}\). The second largest eigenvalue of \(\matp\) is crucial to our algorithms, whose modulus is denoted as \(\lambda\).

For a random walk on digraph \(\gr\), numerous associated quantities can be expressed in terms of the fundamental matrix \(\fund\)~\cite{MeCa75}. For a random walk with transition matrix \(\matp\), the fundamental matrix is defined as the group inverse of \(\mati-\matp\):
\begin{equation*}
      \fund=\spar{\mati-\matp}^\#=\bigspar{\mati-\matp+\vecone\vecpi^\top}^{-1}-\vecone\vecpi^\top.
\end{equation*}
As the generalized inverse, \(\fund\) satisfies \(\spar{\mati-\matp}\fund\spar{\mati-\matp}=\mati-\matp\).
Additionally, we can easily prove that \(\fund\) and \(\mati-\matp\) share the same left null space and right null space. Concretely, \(\fund\vecone=\spar{\mati-\matp}\vecone=\veczero\) and \(\vecpi^\top\fund=\vecpi^\top\spar{\mati-\matp}=\veczero^\top\).

\subsection{Hitting Time and Kemeny's Constant}\label{subsec:hit-kem}

A key concept in random walks is hitting time~\cite{Lo93,CoBeTeVoKl07}. The hitting time \(\hit{i}{j}\) is defined as the expected time for a random walker originating at node \(i\) to arrive at node \(j\) for the first time.
Several important quantities can be derived from the hitting time, here we only consider the absorbing random-walk centrality and Kemeny's constant.

For a node \(s\) in the \(n\)-node digraph \(\gr=(V,E)\), its absorbing random-walk centrality \(\absorb{s}\) is defined as \(\absorb{s}=\sum_{i=1}^n\pi_i\hit{i}{s}\). Lower values of \(\absorb{s}\) indicate higher importance for node \(s\), which has been analyzed extensively~\cite{TeBeVo09,Be09,Be16}. For brevity, we refer to \(\absorb{s}=\sum_{i=1}^{n}\pi_i\hit{i}{s}\) as \textit{walk centrality} henceforth.
For an \(n\)-node digraph \(\gr=(V,E)\), its Kemeny's constant \(\kem\) is defined as the expected steps for a walker starting from node \(u\) to node \(i\) selected with probability \(\pi_i\). Formally, \(\kem=\sum_{i=1}^{n}\pi_i\hit{u}{i}\).
The invariance of Kemeny's constant \(\kem\) stems from the fact that its value remains unchanged regardless of the chosen starting node \(u\).

As discussed in \secref{subsec:graph-rand}, many quantities associated with random walks are determined by the fundamental matrix \(\fund\). For example, the hitting time \(\hit{i}{j}\) satisfies \(\hit{i}{j}=\pi_j^{-1}\spar{\fund_{j,j}-\fund_{i,j}}\)~\cite{MeCa75}. Therefore, the walk centrality can be expressed as
\begin{equation}\label{eq:walk-centr}
      \absorb{s}=\sum_{i=1}^{n}\pi_i\hit{i}{s}=\sum_{i=1}^{n}\frac{\pi_i}{\pi_s}\spar{\fund_{s,s}-\fund_{i,s}}=\frac{\fund_{s,s}}{\pi_s},
\end{equation}
while Kemeny's constant can be represented as
\begin{equation}\label{eq:kem-exact}
      \kem=\sum_{i=1}^{n}\pi_i\hit{u}{i}=\sum_{i=1}^{n}\spar{\fund_{i,i}-\fund_{u,i}}=\normtrace{\fund}.
\end{equation}

The Kemeny constant can be exactly computed by using~\eqref{eq:kem-exact}. However, this formula requires all the diagonal elements of a group inverse. Since the complexity of matrix inversion is \(O(n^3)\), direct computation of Kemeny's constant is impractical for large-scale networks with millions of nodes.

\subsection{Existing Methods}\label{subsec:exist-methods}

\ifthenelse{\boolean{fullver}}{
      In this subsection, we provide an overview of the existing methods for estimating Kemeny's constant. The first existing method \algoname{ApproxKemeny} is based on solving Laplacian systems of linear equations, while the second existing method \algoname{DynamicMC} simulates truncated random walks. Both of these two previous algorithms are only suitable for undirected graphs.
}

\subsubsection{Method based on Laplacian solver}

For an undirected graph, its Kemeny's constant is equal to the trace of \(\normlap^\dagger\), where \(\normlap\) denotes the normalized Laplacian matrix. Using Hutchinson's method~\cite{Hut89}, \algoname{ApproxKemeny} by Xu~\textit{et al.}~\cite{XuShZhKaZh20} reformulates estimating \(\bigtrace{\normlap^\dagger}\) as approximating the quadratic forms of \(\normlap^\dagger\), which is connected to solving linear equations associated with the Laplacian matrix. Leveraging a nearly linear-time Laplacian solver~\cite{KySa16}, \algoname{ApproxKemeny} attains nearly linear-time complexity in terms of edge number. However, as shown in \secref{sec:experim}, the high memory usage of Laplacian solver makes this algorithm impractical for large-scale networks. Additionally, the Laplacian solver inherently leverages specific properties of undirected graphs, precluding its application to \algoname{ApproxKemeny} for digraphs.

\subsubsection{Method based on truncated random walks}

To estimate Kemeny's constant, \algoname{DynamicMC}~\cite{LiHuLe21} simulates truncated random walks starting from each node in the network, and then sums up the probabilities of returning to the source. \algoname{DynamicMC} also incorporates a heuristic strategy, where the simulation terminates once the change of the sum falls below a given threshold. Under the parallel GPU implementation, \algoname{DynamicMC} achieves exceptional performance surpassing previous methods. However, the exceptional performance largely stems from GPU implementation, which is not competitive under fair comparison in \secref{sec:experim}. Additionally, the performance limitations of \algoname{DynamicMC} are two-fold. First, the inclusion of this heuristic strategy lacks a theoretical guarantee for maintaining the accuracy. Second, \algoname{DynamicMC} relies on theoretical analysis for undirected graphs, whose Kemeny's constant is expressed as the infinite sum over powers of \(\lambda_i\). This expression cannot  be extended to digraphs, since the matrix \(\matp\)  for a digraph loses diagonalizability and \(\lambda_i\) becomes complex-valued. Hence, \algoname{DynamicMC} cannot be directly extended to digraphs. In contrast, we address this issue by reformulating Kemeny's constant as the trace of fundamental matrix, which holds for both directed and undirected graphs.

\section{Theoretical Results}\label{sec:theor-res}

\ifthenelse{\boolean{fullver}}{
      In this section, we provide a series of theoretical results for Kemeny's constant. First, we demonstrate that, given the estimation of the stationary distribution of a random walk, the diagonal elements of the fundamental matrix can be well approximated through a truncated sum. This result motivates the design of Monte Carlo estimator for Kemeny's constant based on truncated random walks. Second, we propose an alternative formula relating Kemeny's constant to the walk centrality of a specific node \(s\) and the inverse of a submatrix \(\mati-\matp_{-s}\). This result reformulates the estimation of Kemeny's constant, which prompts us to design Monte Carlo algorithm based on forest sampling.
}

\subsection{Approximation for  Fundamental Matrix by Truncated Sum}

As shown in~\eqref{eq:kem-exact}, the Kemeny constant is intimately related to the diagonal elements of the fundamental matrix \(\fund\). In this subsection, we first demonstrate that the trace of \(\fund\) can be approximated by an \(l\)-truncated sum \(\truncfund\) with an additive error bound. Furthermore, we prove that the diagonal elements of \(\fund\) can also be approximated by \(\truncfund\) with arbitrary accuracy.

\begin{lemma}\label{lem:infin-sum}
      Let \(\gr=(V,E)\) be a digraph with transition matrix \(\matp\) and stationary distribution \(\vecpi\). The fundamental matrix \(\fund\) can be expressed as
      \begin{equation*}
            \fund=\spar{\mati-\matp}^\#=\sum_{k=0}^{\infty}\bigspar{\matp^k-\vecone\vecpi^\top}.
      \end{equation*}
\end{lemma}
\begin{proof}
      According to the Perron-Frobenius theorem~\cite{Bi93}, the \(1\)-eigenvalue is simple. Recall that the left and right \(1\)-eigenvector of \(\matp\) is, respectively, \(\vecpi^\top\) and \(\vecone\). Therefore, the spectral radius of \(\matp-\vecone\vecpi^\top\) is lower than \(1\), and \(\fund\) can be represented as
      \begin{align*}
            \fund & =\spar{\mati-\matp}^\# =\bigspar{\mati-\matp+\vecone\vecpi^\top}^{-1}-\vecone\vecpi^\top                                              \\
                  & =\midpar{\mati-\bigspar{\matp-\vecone\vecpi^\top}}^{-1}-\vecone\vecpi^\top                                                            \\
                  & =-\vecone\vecpi^\top+\sum_{k=0}^{\infty}\bigspar{\matp-\vecone\vecpi^\top}^k=\sum_{k=0}^{\infty}\bigspar{\matp^k-\vecone\vecpi^\top}.
      \end{align*}
      where the last equality can be easily obtained through mathematical induction.
\end{proof}


\lemref{lem:infin-sum} indicates that \(\fund\) can be represented by an infinite sum. Therefore, we attempt to approximate \(\fund\) by an \(l\)-truncated sum, which is defined as \(\truncfund=\sum_{k=0}^{l}\bigspar{\matp^k-\vecone\vecpi^\top}\). We begin by approximating \(\normtrace{\fund}\) by \(\bigtrace{\truncfund}\) with a theoretical error bound.

\begin{lemma}\label{lem:trace-truncate}
      Let \(\gr=(V,E)\) be an \(n\)-node digraph with transition matrix \(\matp\). For any \(\epsilon>0\), if \(l\) is selected satisfying \(l\geq\frac{\log\spar{n^{-1}(\epsilon-\epsilon\lambda)}}{\log\spar{\lambda}}\), then we have
      \ifthenelse{\not\boolean{fullver}}{
            \(\abs{\normtrace{\fund}-\bigtrace{\truncfund}}\leq\epsilon\).
      }{
            \begin{equation*}
                  \abs{\normtrace{\fund}-\bigtrace{\truncfund}}\leq\epsilon.
            \end{equation*}
      }
\end{lemma}
\begin{proof}
    \begin{align*}
        \abs{\trace{\fund}-\bigtrace{\truncfund}}
         & =\abs{\sum_{k=l+1}^{\infty}\spar{\bigtrace{\matp^k}-1}}\leq\sum_{k=l+1}^{\infty}\sum_{i=2}^{n}\abs{\lambda_i}^k \\
         & \leq n\sum_{k=l+1}^{\infty}\lambda^k=\frac{n\lambda^{l+1}}{1-\lambda}\leq{\epsilon}.
    \end{align*}
    This finishes the proof.
\end{proof}


After giving a theoretical bound of \(\bigtrace{\truncfund}\), we next approximate the \(\myord{i}\) diagonal element of \(\fund\) with arbitrary accuracy. This poses a greater challenge, as \(\matp\) is not diagonalizable for digraphs. In order to bound the error introduced by the truncated length \(l\), we need the following lemma.

\begin{lemma}\cite{AlFi02}  \label{lem:dist-deviat}
      For a digraph \(\gr=(V,E)\) with transition matrix \(\matp\) and stationary distribution \(\vecpi\), suppose there exists a probability measure \(\mu\), a real number \(\delta>0\) and  time \(t\) such that \(\matp^t_{i,j}\geq\delta\mu_j\) holds for any nodes \(i,j\in V\). Then for any node \(i\in V\) and integer \(k\geq0\), we have
      \ifthenelse{\not\boolean{fullver}}{
      \(\abs{\matp^k_{i,i}-\pi_i}\leq\spar{1-\delta}^{\floor{\nicefrac{k}{t}}}\).
      }{
      \begin{equation*}
            \abs{\matp^k_{i,i}-\pi_i}\leq\spar{1-\delta}^{\floor{\nicefrac{k}{t}}}.
      \end{equation*}
      }
\end{lemma}

\ifthenelse{\boolean{fullver}}{
Recall that \(\matp^k_{i,i}\) represents the probability of the random walker who starts at \(i\) and returns to \(i\) exactly after \(k\) steps. Therefore, \lemref{lem:dist-deviat} actually bounds the deviation of probability distribution from stationarity.
}

Subsequently, we introduce a theoretical bound of \(\truncfund_{i,i}\):

\begin{lemma}\label{lem:diag-truncate}
      Given an \(n\)-node digraph \(\gr=(V,E)\), let \(\dmax\) be the maximum out-degree, and let \(\tau\) be the diameter of \(\gr\), which is the longest distance between nodes. For any \(\epsilon>0\), if \(l\) is selected satisfying \(l\geq\frac{\tau\log\spar{n\epsilon\tau^{-1}\dmax^{-\tau}}}{\log\spar{1-n\dmax^{-\tau}}}+\tau-1\), then we have
      \ifthenelse{\not\boolean{fullver}}{
            \(\abs{\fund_{i,i}-\truncfund_{i,i}}\leq \epsilon\).
      }{
            \begin{equation*}
                  \abs{\fund_{i,i}-\truncfund_{i,i}}\leq \epsilon.
            \end{equation*}
      }
\end{lemma}
\ifthenelse{\boolean{fullver}}{\begin{proof}
    Since \(\tau\) represents the longest distance between all pairs of nodes, the inequality \(\matp^{\tau}_{i,j}\geq\dmax^{-\tau}\) holds for any nodes \(i,j\in V\). Plugging this into \lemref{lem:dist-deviat}, we obtain
    \begin{align*}
        \abs{\fund_{i,i}-\truncfund_{i,i}}
         & =\abs{\sum_{k=l+1}^{\infty}\matp^k_{i,i}-\pi_i}\leq\sum_{k=l+1}^{\infty}\abs{\matp^k_{i,i}-\pi_i}    \\
         & \leq\sum_{k=l+1}^{\infty}\spar{1-n\dmax^{-\tau}}^{\floor{\nicefrac{k}{\tau}}}                        \\
         & \leq\frac{\tau\spar{1-n\dmax^{-\tau}}^{\floor{\nicefrac{(l+1)}{\tau}}}}{n\dmax^{-\tau}}\leq\epsilon,
    \end{align*}
    where the last inequality is due to \(l\geq\frac{\tau\log\spar{n\epsilon\tau^{-1}\dmax^{-\tau}}}{\log\spar{1-n\dmax^{-\tau}}}+\tau-1\).
\end{proof}
}

\subsection{Alternative Formula for Kemeny's Constant}

In this subsection, we introduce an alternative formula that associates Kemeny's constant of digraphs with a submatrix of \(\mati-\matp\). The enhanced diagonal dominance of this submatrix facilitates more accurate approximation.

\begin{theorem}\label{thm:kem-submat}
      Let \(\gr=(V,E)\) be an \(n\)-node digraph with transition matrix \(\matp\), stationary distribution \(\vecpi\) and the fundamental matrix \(\fund\). For any node \(s\in V\), the Kemeny constant of \(\gr\) can be represented as
      \begin{equation}\label{eq:kem-submat}
            \kem=\bigtrace{\subfund}-\absorb{s}.
      \end{equation}
\end{theorem}
\begin{proof}
    As shown in \secref{subsec:graph-rand}, the left and right null vectors of \(\fund\) may be inequivalent. This prompts us to use the diagonally scaled fundamental matrix \(\normfund\). \(\normfund\) is defined as \(\normfund=\matpi^{\nicehalf}\fund\matpi^{-\nicehalf}\), where \(\matpi=\diag{\pi_1,\pi_2,\dots,\pi_n}\). Similarly, we define the diagonally scaled Laplacian as \(\diaglap=\matpi^{\nicehalf}\spar{\mati-\matp}\matpi^{-\nicehalf}\). It is easy to verify that \(\normfund\) and \(\diaglap\) share the same left and right null vector \(\vecpi^{\nicehalf}\).

    Subsequently, we attempt to establish the connection between \(\normfund_{-s}\) and \(\diaglap_{-s}^{-1}\). After properly adjusting the node labels, \(\normfund\) and \(\diaglap\) can be rewritten in block forms as
    \begin{equation*}
        \normfund=
        \begin{pmatrix}
            \normfund_{-s}  & \normfund_{T,s} \\
            \normfund_{s,T} & \normfund_{s,s}
        \end{pmatrix},
        \diaglap=
        \begin{pmatrix}
            \diaglap_{-s}  & \diaglap_{T,s} \\
            \diaglap_{s,T} & \diaglap_{s,s}
        \end{pmatrix},
    \end{equation*}
    where \(T=V\setminus\setof{s}\). If \(\normvecpi\) is denoted as \(\pi_s^{-\nicehalf}\vecpi^{\nicehalf}_{-s}\in\rea^{n-1}\), then we prove that \(\matx=\spar{\mati\enspace {-\normvecpi}}\normfund\begin{pmatrix}\mati\\-\normvecpi^\top\end{pmatrix}\) equals to \(\diaglap_{-s}^{-1}\):
    \begin{align*}
        \diaglap_{-s}\matx\diaglap_{-s}
         & =\diaglap_{-s}\spar{\mati\enspace {-\normvecpi}}\normfund\begin{pmatrix}\mati\\-\normvecpi^\top\end{pmatrix}\diaglap_{-s} \\
         & =\spar{\diaglap_{-s}\ \diaglap_{T,s}}\normfund\begin{pmatrix}\diaglap_{-s}\\\diaglap_{s,T}\end{pmatrix},
    \end{align*}
    where the last equality is due to \(\spar{\normvecpi\enspace 1}\diaglap=\veczero^\top\) and \(\diaglap\begin{pmatrix}\normvecpi\\1\end{pmatrix}=\veczero\).
    Then, we obtain
    \begin{align*}
        \diaglap_{-s}\matx\diaglap_{-s}
         & =\spar{\diaglap_{-s}\ \diaglap_{T,s}}\normfund\begin{pmatrix}\diaglap_{-s}\\\diaglap_{s,T}\end{pmatrix}=\spar{\mati\enspace \veczero}\diaglap\normfund\diaglap\begin{pmatrix}\mati\\\veczero^\top\end{pmatrix} \\
         & =\spar{\mati\enspace \veczero}\diaglap\begin{pmatrix}\mati\\\veczero^\top\end{pmatrix}=\diaglap_{-s}.
    \end{align*}
    Finally, we are able to prove Equation~\eqref{eq:kem-submat} as
    \begin{align*}
          & \bigtrace{\subfund}=\bigtrace{\diaglap_{-s}^{-1}}=\trace{\matx}                                                                \\
        = & \bigtrace{\normfund_{-s}}-\normvecpi^\top\normfund_{T,s}-\normfund_{s,T}\normvecpi+\normfund_{s,s}\normvecpi^\top\normvecpi    \\
        = & \bigtrace{\normfund_{-s}}+\normfund_{s,s}\spar{\normvecpi^\top\normvecpi+2}=\bigtrace{\normfund}+\frac{\normfund_{s,s}}{\pi_s} \\
        = & \trace{\fund}+\frac{\fund_{s,s}}{\pi_s}=\kem+\absorb{s}.
    \end{align*}
    Here the fourth equality is due to \(\spar{\normvecpi\enspace 1}\normfund=\veczero^\top\) and \(\normfund\begin{pmatrix}\normvecpi\\1\end{pmatrix}=\veczero\).
\end{proof}


\thmref{thm:kem-submat} reveals that for any selected node \(s\), the estimation of Kemeny's constant boils down to the evaluation of the trace of \(\subfund\) and the walk centrality of \(s\). This alternative formula motivates us to design an approximation algorithm that estimates these two quantities separately.

\section{Algorithm Design}\label{sec:alg-design}

\ifthenelse{\boolean{fullver}}{
      After obtaining theoretical results for Kemeny's constant of digraphs, we next utilize them to design efficient algorithms to approximate it. In this section, we first propose a Monte Carlo algorithm with the help of \lemref{lem:infin-sum}. This algorithm is based on truncated random walks with several optimization techniques. Additionally, we present another Monte Carlo algorithm with the help of \thmref{thm:kem-submat}. This algorithm randomly samples directed rooted spanning trees of given digraph. After describing each approximate algorithm, we give the analysis of its time complexity and theoretical accuracy.
}

\subsection{Truncated Random Walk Based  Algorithm }\label{subsec:alg-trunc}

Combining~\eqref{eq:kem-exact} with \lemref{lem:trace-truncate}, we can approximate Kemeny's constant of a digraph with an \(l\)-truncated sum, which is defined as
\begin{equation}\label{eq:def-trunckem}
      \trunckem=\bigtrace{\truncfund}=\sum_{k=0}^{l}\spar{\bigtrace{\matp^k}-1}.
\end{equation}
We note that this estimator is the same as that in \algoname{DynamicMC}~\cite{LiHuLe21}, but is proposed from a different perspective. Equation~\eqref{eq:def-trunckem} indicates that \algoname{DynamicMC} actually supports digraphs. While the analysis of \algoname{DynamicMC} is restricted to undirected graphs, the approximation error of \(\trunckem\) for digraphs is analyzed in \lemref{lem:trace-truncate}. Since the direct computation of \(\matp^k_{i,i}\) involves time-consuming matrix multiplication, we resort to a Monte Carlo approach similar to \algoname{DynamicMC}. Concretely, for each node \(i\in V\), we simulate \(r\) independent \(l\)-truncated random walks starting from \(i\). Let \(\retnum{i}{j}\) denote the times of \(\myord{j}\) random walk that return to \(i\), and let \(\meanretnum{i}\) denote the empirical mean of \(\retnum{i}{j}\), then we can give an unbiased estimator of \(\trunckem\) based on \(\meanretnum{i}\), which is defined as \(\approxkem=n-l-1+\sum_{i=1}^{n}\meanretnum{i}\).

Although \(\approxkem\) is unbiased, we have to simulate numerous truncated random walks to reduce the approximation error. The loose theoretical error bound leads to excessive simulation, making the approximation algorithm inefficient in practice. Therefore, we utilize several optimization techniques. Recall that the truncated random walk is simulated through two steps: the iteration over nodes and the simulation of random walks starting from each node. Below we make improvements from these two perspectives.

\subsubsection{Adaptive simulation from each node}

In the analysis of \algoname{DynamicMC}, Hoeffding's inequality is utilized to derive a theoretical bound.

\begin{lemma}[Hoeffding's inequality]\label{lem:hoeffding}
      Let \(x_1,x_2,\dots,x_n\) be \(n\) independent random variables such that \(a\leq x_i\leq b\) for \(i=1,2,\dots,n\). Let \(x=\sum_{i=1}^{n}x_i\), then for any \(\epsilon>0\),
      \begin{equation*}
            \Pr\spar{\abs{x-\mean{x}}\geq\epsilon}\leq2\exp\setof{-\frac{2\epsilon^2}{n\spar{b-a}^2}}.
      \end{equation*}
\end{lemma}

However, as Hoeffding's inequality does not consider the variance of random variables, the provided theoretical bound tends to be relatively loose. To address this limitation, we leverage the empirical Bernstein inequality~\cite{HuSeTa07}:

\begin{lemma}\label{lem:bernstein}
      Let \(X_1,X_2,\dots,X_n\) be \(n\) real-valued i.i.d. random variables that satisfy \(0\leq X_i\leq X_{\mathrm{sup}}\). If we denote \(\bar{X}\) and \(X_{\mathrm{var}}\) as the empirical mean and the empirical variance of \(X_i\), then we have
      \begin{equation*}
            \Pr\spar{\abs{\bar{X}-\mathbb{E}[\bar{X}]}\geq \bernfunc{n,X_{\mathrm{var}},X_{\mathrm{sup}},\delta}}\leq\delta,
      \end{equation*}
      where
      \begin{equation*}
            \bernfunc{n,X_{\mathrm{var}},X_{\mathrm{sup}},\delta}=\sqrt{\frac{2X_{\mathrm{var}}\log\spar{\nicefrac{3}{\delta}}}{n}}+\frac{3X_{\mathrm{sup}}\log\spar{\nicefrac{3}{\delta}}}{n}.
      \end{equation*}
\end{lemma}

\lemref{lem:bernstein} differs from \lemref{lem:hoeffding} in that it involves the empirical variance of random variables. While the empirical variance remains unknown a priori, it can be efficiently maintained throughout the simulation. Therefore, our first improvement implements the empirical Bernstein inequality, but still retains the Hoeffding bound to preserve theoretical accuracy. Meanwhile, if the empirical error of \(\meanretnum{i}\) provided by \lemref{lem:bernstein} falls below the threshold, we terminate the simulation of random walks starting from \(i\). Crucially, the theoretical accuracy of estimating Kemeny's constant remains unaffected by applying this adaptive strategy.

\subsubsection{Iteration over node subset}

\ifthenelse{\boolean{fullver}}{
      Despite the refinement by using \lemref{lem:bernstein}, the performance of MC algorithm remains non-optimal for large digraphs. This inefficiency stems from the simulation from every node in digraph \(\gr\) as a starting point, which is time-consuming.
}

Note that Kemeny's constant of \(\gr\) is concerned with the sum of \(\meanretnum{i}\), we attempt to estimate it by only summing a small proportion of \(\meanretnum{i}\), which significantly reduces the required number of simulated random walks.

Specifically, for an \(n\)-node digraph \(\gr=(V,E)\), a node subset \(\sqrtset\subseteq V\) of capacity \(k\ll n\) is sampled uniformly at random. To efficiently estimate Kemeny's constant, the original sum \(S=\sum_{u\in V}\meanretnum{u}\) is replaced by the partial sum \(\tilde{S}=\nicefrac{n}{k}\sum_{i\in\sqrtset}\meanretnum{i}\). As an unbiased estimator of \(S\), \(\tilde{S}\) is also suitable for estimating Kemeny's constant.

\begin{lemma}\label{lem:sqrt-sampling-unbiased}
      \(\tilde{S}\) is an unbiased estimator for \(S\).
\end{lemma}
\ifthenelse{\boolean{fullver}}{\begin{proof}
    For a strongly connected \(n\)-node digraph \(\gr=(V,E)\), we denote the set of \(k\)-combinations of \(V\) as \(V_k\). Since the node subset \(\sqrtset\) is sampled uniformly at random, we can represent the expected value of \(\tilde{S}\) as
    \begin{align*}
        \bigmean{\tilde{S}}
         & =\sum_{\sqrtset\in V_k}\frac{1}{\binom{n}{k}}\frac{n}{k}\sum_{i\in\sqrtset}\meanretnum{i}                    \\
         & =\frac{1}{\binom{n}{k}}\frac{n}{k}\sum_{u\in V}\binom{n-1}{k-1}\meanretnum{u}=\sum_{u\in V}\meanretnum{u}=S,
    \end{align*}
    which completes our proof.
\end{proof}
}

Next, we provide a guarantee for the additive error of \(\tilde{S}\):

\begin{lemma}\label{lem:sqrt-sampling-error}
      Given \(n\) positive numbers \(x_1,x_2,\dots,x_n\in[0,b]\) with their sum \(x=\sum_{i=1}^{n}x_i\) and an error parameter \(\epsilon>n^{-\nicehalf}\log^{\nicehalf}\spar{2n}\). If we randomly select \(c=\bigbigo{b\epsilon^{-1}n^{\nicehalf}\log^{\nicehalf}n}\) numbers, \(x_{v_1},x_{v_2},\dots,x_{v_c}\) by Bernoulli trials with success probability \(p=b\epsilon^{-1}n^{-\nicehalf}\log^{\nicehalf}\spar{2n}\) satisfying \(0<p<1\), and define \(\tilde{x}=\sum_{i=1}^{c}x_{v_i}/p\), then \(\tilde{x}\) is an approximation for the sum \(x\) of the original \(n\) numbers, satisfying \(\abs{x-\tilde{x}}\leq n\epsilon\).
\end{lemma}
\ifthenelse{\boolean{fullver}}{\begin{proof}
    For \(i=1,2,\dots,n\), let \(y_i\) be Bernoulli random variable such that \(\Pr\spar{y_i=1}=p\) and \(\Pr\spar{y_i=0}=1-p\). Here \(y_i=1\) indicates that \(x_i\) is selected and \(y_i=0\) otherwise. Let \(z_i=x_iy_i\) be \(n\) independent random variables with \(0\leq z_i\leq b\). Denote the sum of random variables \(y_i\) as \(y\), and denote the sum of random variables \(z_i\) as \(z\). Namely, \(y=\sum_{i=1}^{n}y_i\) and \(z=\sum_{i=1}^{n}z_i\). It is clear that \(y\) and \(z\) represent the number of selected numbers and their sum. Therefore the expectations of \(y\) and \(z\) can be expressed as \(\mean{y}=np\) and \(\mean{z}=px\). According to \lemref{lem:hoeffding}, we have
    \begin{equation*}
        \Pr\spar{\abs{\tilde{x}-x}\geq n\epsilon}=\Pr\spar{\frac{\abs{z-px}}{p}\geq n\epsilon}\leq2\exp\setof{-\frac{2p^2n^2\epsilon^2}{nb^2}}\leq\frac{1}{n},
    \end{equation*}
    finishing the proof.
\end{proof}
}

\subsubsection{Improved Monte Carlo algorithm}

Using Lemmas~\ref{lem:bernstein} and~\ref{lem:sqrt-sampling-error}, we propose an improved MC algorithm for approximating Kemeny's constant of digraphs, which is depicted in \algoref{algo:improvedmc}.

\begin{algorithm}
    \caption{\algoname{ImprovedMC}\((\gr,\epsilon)\)}
    \label{algo:improvedmc}
    \Input{
        \(\gr=(V,E)\): a digraph with \(n\) nodes;
        \(\epsilon>0\): an error parameter
    }
    \Output{
        \(\sqrtkem\): approximation of Kemeny's constant \(\kem\) in \(\gr\)
    }
    \(l\gets\ceil{\frac{\log\spar{\nicefrac{3}{(\epsilon-\epsilon\lambda)}}}{\log\spar{\nicefrac{1}{\lambda}}}}\), \(r\gets\ceil{9\epsilon^{-2}l^2\log\spar{2n}/4}\)\;
    Sample a node set \(\sqrtset\subseteq V\) satisfying \(\abs{\sqrtset}=\min\setof{\ceil{3\epsilon^{-1}ln^{\nicehalf}\log^{\nicehalf}n/2},n}\)\;\label{line:sqrt-sample}
    \ForEach{node \(i\in \sqrtset\)}{\label{line:simwalk-begin}
    \(\meanretnum{i}\gets0\)\;
    \For{\(j=1,2,\dots,r\)}{
    Sample the \(\myord{j}\) \(l\)-truncated random walk starting from \(i\), and let \(\retnum{i}{j}\) be the times of walk that return to \(i\)\;
    \(\meanretnum{i}\gets\meanretnum{i}+\retnum{i}{j}\)\;
    \lIf{\(\bigbernfunc{j,\hat{t}^{(l)}_{\mathrm{var}},\nicefrac{l}{2},\nicefrac{1}{n}}\leq \nicefrac{n\epsilon}{3}\)}{break}\label{line:bernstein}
    }
    \(\meanretnum{i}\gets\meanretnum{i}/j\)\;\label{line:simwalk-end}
    }
    \Return{\(\sqrtkem=n-l-1+\frac{n}{\abs{\sqrtset}}\sum_{i\in\sqrtset}\meanretnum{i}\)}
\end{algorithm}

According to \lemref{lem:sqrt-sampling-error}, \algoname{ImprovedMC} randomly selects \(\abs{\sqrtset}=\bigbigo{\epsilon^{-1}ln^{\nicehalf}\log^{\nicehalf}n}\) nodes from \(V\) (\lineref{line:sqrt-sample}).
\algoname{ImprovedMC} then simulates \(\bigo{\epsilon^{-2}l^2\log n}\) \(l\)-truncated random walks from each selected node \(i\) (\linerangeref{line:simwalk-begin}{line:simwalk-end}). If the empirical error computed by \lemref{lem:bernstein} is less than threshold \(\nicefrac{\epsilon}{3}\), the simulation terminates (\lineref{line:bernstein}). By summing up \(\meanretnum{i}\), \algoname{ImprovedMC} returns \(\sqrtkem\) as the approximation for Kemeny's constant. The performance of \algoname{ImprovedMC} is characterized in \thmref{thm:performance-improvedmc}.

\begin{theorem}\label{thm:performance-improvedmc}
      For an \(n\)-node digraph \(\gr\) and an error parameter \(\epsilon>0\), \algoref{algo:improvedmc} runs in \(\bigbigo{\epsilon^{-3}l^4n^{\nicehalf}\log^{\nicefrac{3}{2}} n}\) time, and returns \(\sqrtkem\) as the approximation for Kemeny's constant \(\kem\) of \(\gr\), which is guaranteed to satisfy \(\bigabs{\sqrtkem-\kem}\leq 2\epsilon\kem\) with high probability.
\end{theorem}
\ifthenelse{\boolean{fullver}}{\begin{proof}
    Deriving the time complexity of \algoname{ImprovedMC} is straightforward. Therefore, our main focus is to provide the relative error guarantee for this algorithm.

    Referring back to~\eqref{eq:def-trunckem}, \(\kem\) is initally approximated by the truncated sum \(\trunckem\). The approximation error of \(\trunckem\) dependent on truncated length \(l\) is analyzed in \lemref{lem:trace-truncate}.
    Furthermore, the estimator \(\approxkem=n-l-1+\sum_{i=1}^{n}\meanretnum{i}\) is leveraged to approximate \(\trunckem\).
    We next provide the connection between \(r\) and the error of \(\approxkem\). According to \lemref{lem:hoeffding}, we have
    \begin{align*}
        \Pr\spar{\abs{\meanretnum{i}-\bigmean{\meanretnum{i}}}\geq\frac{\epsilon}{3}}
         & = \Pr\biggspar{\biggabs{\sum_{j=1}^{r}\retnum{i}{j}-\bigmean{\sum_{j=1}^{r}\retnum{i}{j}}}\geq\frac{r\epsilon}{3}} \\
         & \leq 2\exp\setof{-\frac{2r^2\epsilon^2}{9r\spar{\nicefrac{l}{2}}^2}}\leq\frac{1}{2n^2}.
    \end{align*}
    Based on the union bound, it holds that
    \begin{equation}\label{eq:err-approxkem}
        \abs{\approxkem-\trunckem}\leq\sum_{i=1}^{n}\abs{\meanretnum{i}-\sum_{k=1}^{l}\matp^k_{[i,i]}}\leq\frac{n\epsilon}{3}
    \end{equation}
    with probability
    \begin{equation*}
        \biggspar{1-\frac{1}{n^2}}^n\geq1-n\cdot\frac{1}{2n^2}=1-\frac{1}{2n}.
    \end{equation*}
    As stated in \algoref{algo:improvedmc}, applying \lemref{lem:bernstein} does not introduce additional error since the error of \(\approxkem\) is lower than \(\nicefrac{n\epsilon}{3}\). Therefore, we utilize \lemref{lem:sqrt-sampling-error} to give a theoretical bound for the approximation error of the partial sum \(\sqrtkem\) as
    \begin{equation}\label{eq:err-sqrtkem}
        \Pr\bigspar{\bigabs{\sqrtkem-\approxkem}\geq\frac{n\epsilon}{3}}\leq2\exp\setof{-\frac{2\abs{\sqrtset}\epsilon^2}{n\spar{\nicefrac{l}{2}}^2}}\leq\frac{1}{2n}.
    \end{equation}
    Plugging~\eqref{eq:err-approxkem} and~\eqref{eq:err-sqrtkem} into \lemref{lem:trace-truncate}, we derive the additive error guarantee of \(\sqrtkem\):
    \begin{equation*}
        \Pr\bigspar{\bigabs{\sqrtkem-\kem}\leq n\epsilon}\geq\bigspar{1-\frac{1}{2n}}^2\geq1-\frac{1}{n}.
    \end{equation*}
    As shown by~\cite{Hu06}, the minimum Kemeny's constant across all \(n\)-node digraphs is \(\nicefrac{(n+1)}{2}\). Leveraging this result, we can translate the additive error bound for \(\sqrtkem\) into a relative error guarantee:
    \begin{equation*}
        \Pr\bigspar{\bigabs{\sqrtkem-\kem}\leq 2\epsilon\kem}\geq\Pr\bigspar{\bigabs{\sqrtkem-\kem}\leq n\epsilon}\geq1-\frac{1}{n},
    \end{equation*}
    which completes our proof.
\end{proof}
}

Although the time complexity of \algoname{ImprovedMC} is sublinear with respect to the number of nodes, this complexity bound remains relatively loose due to the inclusion of \lemref{lem:bernstein}.

\subsection{Algorithm Based on Directed Tree Sampling}\label{subsec:alg-tree}

Although \algoname{ImprovedMC} achieves enhanced efficiency through optimization techniques, its accuracy remains to be improved. Leveraging the alternative formula in \thmref{thm:kem-submat}, we propose another MC algorithm \algoname{TreeMC}, which samples incoming directed rooted spanning trees. Due to the improved diagonal dominance of the submatrix in~\eqref{eq:kem-submat}, \algoname{TreeMC} attains enhanced accuracy. After presenting \algoname{TreeMC}, we analyze its time complexity and error guarantee.

Recall from \thmref{thm:kem-submat} that the calculation of Kemeny's constant is reduced to the evaluation of \(\bigtrace{\subfund}\) and \(\absorb{s}\). Meanwhile, \lemref{lem:diag-truncate} and~\eqref{eq:walk-centr} indicate that \(\absorb{s}\) can be estimated by simulating truncated random walks. Therefore, the main goal of our MC algorithm is to approximate \(\bigtrace{\subfund}\).

For random walks in digraph \(\gr=(V,E)\) with absorbing node \(s\), \(\subfund_{i,j}\) denotes the expected passage times over node \(j\) by a random walker initialized at node \(i\) prior to absorption at node \(s\)~\cite{ZhYaLi12}. Using this physical meaning, we can estimate \(\bigtrace{\subfund}\) by simulating absorbing random walks from each node in \(V\setminus\setof{s}\).
However, the expected running time of single simulation is \(\sum_{i=1}^{n}\sum_{j=1}^{n}\subfund_{i,j}=\vecone^\top\subfund\vecone\), which is large due to the dense property of \(\subfund\). To improve the efficiency of simulation, we introduce the method of sampling incoming directed rooted spanning trees, which is essentially simulating loop-erased random walks.

Starting from a node, the loop-erased random walk consists of two phases: simulation of a random walk as well as the erasure of the loop within walk path. The loop-erased random walk is used in Wilson's algorithm~\cite{Wi96}, where loop-erased paths are iteratively generated. The aggregate of these loop-erased paths is a directed, cycle-free subgraph of \(\gr\). In this subgraph, every non-root node \(i\in V\setminus\setof{s}\) has out-degree \(1\), while the root node \(s\) has out-degree \(0\). This subgraph is denoted as an incoming directed spanning tree rooted at \(s\). Based on a variant of Wilson's algorithm, we prove that the diagonal elements of \(\subfund\) can be well approximated.

\begin{lemma}\label{lem:tree-sampling-unbiased}
      For a digraph \(\gr=(V,E)\) with transition matrix \(\matp\), we simulate loop-erased random walks with an absorbing node \(s\). Let \(t_{i}\) denote the passage times over node \(i\), then we have
      \begin{equation}\label{eq:tree-sampling-unbiased}
            \mean{t_i}=\subfund_{i,i}.
      \end{equation}
\end{lemma}
\ifthenelse{\boolean{fullver}}{\begin{proof}
    In the initial round of the loop-erased random walk that starts from node \(i\), there is only one absorbing node \(s\) and the expected passage times over \(i\) is \(\subfund_{i,i}\). Hence~\eqref{eq:tree-sampling-unbiased} holds true for the initial starting node. For Wilson's algorithm, the distribution of sampled random walk path is independent of the node ordering~\cite{Wi96}. Therefore, every node in \(V\setminus\setof{s}\) can be sampled in the first round, which indicates that~\eqref{eq:tree-sampling-unbiased} holds true for all \(i\in V\setminus\setof{s}\).
\end{proof}
}

\lemref{lem:tree-sampling-unbiased} reveals that estimating \(\bigtrace{\subfund}\) via a single sampled spanning tree is equivalent to simulating \(n\) absorbing random walks. However, spanning tree sampling proves substantially more efficient than absorbing walk simulation. This efficiency gain arises because loop-erased walks terminate upon revisiting prior loop-free trajectories. Consequently, we estimate \(\bigtrace{\subfund}\) through sampling spanning trees rather than absorbing walk simulations.

Specifically, we sample \(r\) incoming directed rooted spanning trees by performing loop-erased random walks. Let \(\passnum{i}{j}\) denote the passage times over \(i\) for the \(\myord{j}\) sample, and \(\meanpassnum{i}\) denote the empirical mean over \(j\). Then \(\sum_{i\in V\setminus\setof{s}}\meanpassnum{i}\) is an unbiased estimator of \(\bigtrace{\subfund}\).
We continue to provide an error bound for the sample size \(r\). To this end, we bound the passage times \(\passnum{i}{j}\) with high probability.

\begin{lemma}\label{lem:passnum-upperbound}
      Given an \(n\)-node digraph \(\gr=(V,E)\) and an absorbing node \(s\in V\), let \(\dmax\) be the maximum out-degree, and let \(\tau\) be the diameter of \(\gr\), which is the longest distance between all pairs of nodes. If \(t\) is selected satisfying \(t\geq\rme\tau \dmax^\tau\ceil{\log\spar{4n^2}}/2\), then
      \begin{equation*}
            \Pr\spar{\passnum{i}{j}>t}\leq\frac{1}{4n^2}.
      \end{equation*}
\end{lemma}
\ifthenelse{\boolean{fullver}}{\begin{proof}
    Since \(\passnum{i}{j}\) denotes the passage times of node \(i\) for the \(\myord{j}\) loop-erased random walk, it is easy to verify that \(\passnum{i}{j}\leq\randhit{i}{s}/2\), where \(\randhit{i}{s}\) denotes the hitting time from \(i\) to \(s\) for the \(\myord{j}\) loop-erased random walk. Therefore, we turn to bound \(\randhit{i}{s}\) with high probability, which requires us to provide an upper bound for its expected value \(\hit{i}{s}\).

    Recall that \(\hit{i}{s}\) can be expressed as \(\vece_i^\top\subfund\vecone\), we have
    \begin{equation}\label{eq:hitting-upperbound}
        \begin{split}
            \max_{i\in V\setminus\setof{s}}\hit{i}{s}
             & =\Abs{\subfund\vecone}_\infty=\sum_{k=0}^{\infty}\Abs{\matp^k_{-s}}_\infty                                                        \\
             & \leq\sum_{k=0}^{\infty}\spar{1-\dmax^{-\tau}}^{\floor{\nicefrac{k}{\tau}}}=\frac{\tau}{1-\spar{1-\dmax^{-\tau}}}=\tau \dmax^\tau.
        \end{split}
    \end{equation}
    According to~\cite{AlFi02}, we can finish our proof by providing the upper bound for \(\randhit{i}{s}\) with high probability:
    \begin{align*}
        \Pr\spar{\passnum{i}{j}>t}
         & =\Pr\spar{\randhit{i}{s}>2t}\leq\exp\setof{-\floor{\frac{2t}{\rme\hit{i}{s}}}} \\
         & =\exp\setof{-\floor{\frac{2t}{\rme\tau \dmax^\tau}}}\leq\frac{1}{4n^2},
    \end{align*}
    where the last inequality is due to \(t\geq\rme\tau \dmax^\tau\ceil{\log\spar{4n^2}}/2\).
\end{proof}
}

Finally, we introduce a theoretical bound for the sample size \(r\).

\begin{lemma}\label{lem:trace-approx}
      Given an \(n\)-node digraph \(\gr=(V,E)\) with absorbing node \(s\), if we sample \(r\) incoming directed rooted spanning trees satisfying \(r\geq\epsilon^{-2}\rme^2\tau^2 \dmax^{2\tau}\ceil{\log^3\spar{4n^2}}\), then for any \(\epsilon>0\), we have
      \begin{equation*}
            \Pr\biggspar{\biggabs{\bigtrace{\subfund}-\sum_{i\in V\setminus\setof{s}}\meanpassnum{i}}\geq\frac{n\epsilon}{2}}\leq\frac{1}{2n}.
      \end{equation*}
\end{lemma}
\ifthenelse{\boolean{fullver}}{\begin{proof}
    \lemref{lem:passnum-upperbound} reveals that \(\passnum{i}{j}\) exhibits an explicit upper bound \(t=\rme\tau \dmax^\tau\ceil{\log\spar{4n^2}}/2\) with a high probability. Plugging this into \lemref{lem:hoeffding}, we obtain
    \begin{align*}
             & \Pr\spar{\abs{\meanpassnum{i}-\mean{\meanpassnum{i}}}\geq\frac{\epsilon}{2}}                                       \\
        =    & \Pr\biggspar{\biggabs{\sum_{j=1}^{r}\passnum{i}{j}-\bigmean{\sum_{j=1}^{r}\passnum{i}{j}}}\geq\frac{r\epsilon}{2}} \\
        \leq & 1-\spar{1-2\exp\setof{-\frac{2r^2\epsilon^2}{4rt^2}}}\spar{1-\frac{1}{4n^2}}                                       \\
        \leq & 1-\spar{1-\frac{1}{4n^2}}^2\leq1-\spar{1-\frac{2}{4n^2}}=\frac{1}{2n^2},
    \end{align*}
    where the second inequality follows from \(r\geq4\epsilon^{-2}t^2\log\spar{2n^2}\). Based on the union bound, it holds that
    \begin{align*}
             & \biggabs{\bigtrace{\subfund}-\sum_{i\in V\setminus\setof{s}}\meanpassnum{i}}               \\
        \leq & \sum_{i\in V\setminus\setof{s}}\abs{\subfund_{i,i}-\meanpassnum{i}}\leq\frac{n\epsilon}{2}
    \end{align*}
    with probability
    \begin{equation*}
        \spar{1-\frac{1}{2n^2}}^n\geq1-n\cdot\frac{1}{2n^2}=1-\frac{1}{2n}.
    \end{equation*}
    This finishes the proof.
\end{proof}
}

Based on the above analyses, we propose a more accurate MC algorithm \algoname{TreeMC} for estimating Kemeny's constant, which is depicted in \algoref{algo:treemc}. Given an \(n\)-node digraph \(\gr\) and an error parameter \(\epsilon\), By selecting absorbing node \(s\) as the node with the largest \(\pi_s\), \algoname{TreeMC} first estimates \(\bigtrace{\subfund}\) (\linerangeref{line:lewalk-begin}{line:lewalk-end}). For this purpose, \algoname{TreeMC} samples \(\bigbigo{\epsilon^{-2}\tau^2\dmax^{2\tau}\log^3n}\) incoming directed rooted spanning trees. The sampling procedure consists of the random walk part (\linerangeref{line:lewalk-begin}{line:lewalk-end}) and the loop-erasure part (\linerangeref{line:loop-erase-begin}{line:loop-erase-end}). Then, \algoname{TreeMC} estimates \(\absorb{s}\) (\linerangeref{line:approx-diag-begin}{line:approx-diag-end}). Analogous to \algoref{algo:improvedmc}, \algoname{TreeMC} also simulates \(l\)-truncated random walks and introduces \lemref{lem:bernstein} for optimization. Combining these two estimations, \algoname{TreeMC} returns \(\treekem\) as the approximation for Kemeny's constant. The performance of \algoname{TreeMC} is analyzed in \thmref{thm:performance-treemc}. Again, the additive error guarantee of \(n\epsilon\) can be converted to the relative error guarantee of \(2\epsilon\) due to the minimum Kemeny's constant of \(n\)-node digraph being \(\nicefrac{(n+1)}{2}\)~\cite{Hu06}.

\begin{algorithm}
    \caption{\algoname{TreeMC}\((\gr,\epsilon)\)}
    \label{algo:treemc}
    \Input{
        \(\gr=(V,E)\): a digraph with \(n\) nodes;
        \(\epsilon\): an error parameter
    }
    \Output{
        \(\treekem\): approximation of Kemeny's constant \(\kem\) in \(\gr\)
    }
    \(s\gets\argmax_{i\in V}\pi_i\), \(r\gets\ceil{\epsilon^{-2}\rme^2\tau^2 \dmax^{2\tau}\ceil{\log^3\spar{4n^2}}}\)\;\label{line:approx-trace-begin}
    \For{\(j=1,2,\dots,r\)}{
        \(\passnum{i}{j}\gets0\) for \(i\in V\setminus\setof{s}\)\;
        \(\text{next}_i\gets\text{arbitrary value}\) for \(i\in V\setminus\setof{s}\)\;
        \(\text{InTree}_i\gets\text{false}\) for \(i\in V\setminus\setof{s}\)\;
        \(\text{InTree}_s\gets\text{true}\)\;
        \ForEach{\(u\in V\setminus\setof{s}\)}{
            \(i\gets u\)\;\label{line:lewalk-begin}
            \While{\(\text{InTree}_i=\text{false}\)}{
                \(\passnum{i}{j}\gets\passnum{i}{j}+1\)\;
                \(\text{next}_i\gets\) a randomly selected out-neighbor of \(i\)\;
                \(i\gets\text{next}_i\)\;\label{line:lewalk-end}
            }
            \(i\gets u\)\;\label{line:loop-erase-begin}
            \While{\(\text{InTree}_i=\text{false}\)}{
                \(\text{InTree}_i\gets\text{true}\), \(i\gets\text{next}_i\)\;\label{line:loop-erase-end}
            }
        }
    }
    \(\meanpassnum{i}\gets\sum_{j=1}^{r}\passnum{i}{j}/r\) for \(i\in V\setminus\setof{s}\)\;\label{line:approx-trace-end}

    \(l\gets\ceil{\frac{\tau\log\spar{n\epsilon(2\tau)^{-1}\dmax^{-\tau}}}{\log\spar{1-n\dmax^{-\tau}}}}+\tau-1\), \(r'\gets\ceil{\frac{l^2\log n}{2\epsilon^2\pi_s^2n^2}}\),\(\meanretnum{s}\gets0\)\;\label{line:approx-diag-begin}
    \For{\(j=1,2,\dots,r'\)}{
    Sample the \(\myord{j}\) \(l\)-truncated random walk starting from \(s\), and let \(\retnum{s}{j}\) be the times of walk returning to \(s\)\;
    \(\meanretnum{s}\gets\meanretnum{s}+\retnum{s}{j}\)\;
    \lIf{\(\bigbernfunc{j,\hat{t}^{(l)}_{\mathrm{var}},\nicefrac{l}{2},\nicefrac{1}{n}}\leq \nicefrac{\sqrt{n}\epsilon}{2}\)}{break}
    }
    \(\meanretnum{s}\gets\meanretnum{s}/j\)\;\label{line:approx-diag-end}
    \Return{\(\treekem=-\pi_s^{-1}\bigspar{\meanretnum{s}+1}+\sum_{i\in V\setminus\setof{s}}\meanpassnum{i}\)}
\end{algorithm}

\begin{theorem}\label{thm:performance-treemc}
      For an \(n\)-node digraph \(\gr=(V,E)\) with absorbing node \(s\) and transition matrix \(\matp\), let \(\tau\) denote the diameter of \(\gr\) and let \(\dmax\) denote the maximum out-degree of nodes in \(\gr\). If the error parameter \(\epsilon\) is determined, then the time complexity of \algoref{algo:treemc} is \(\bigo{T}\), where
      \begin{equation*}
            T=\epsilon^{-2}\tau^2\dmax^{2\tau}\log^3n\cdot\bigtrace{\subfund}+\frac{l^3\log n}{2\epsilon^2\pi_s^2n^2}.
      \end{equation*}
      \algoref{algo:treemc} returns \(\treekem\) as the approximation for Kemeny's constant \(\kem\) of \(\gr\), which satisfies \(\abs{\treekem-\kem}\leq 2\epsilon\kem\).
\end{theorem}
\ifthenelse{\boolean{fullver}}{\begin{proof}
    The theoretical proof of the additive error guarantee for \algoname{TreeMC} is straightforward by plugging \lemref{lem:trace-approx} and \lemref{lem:diag-truncate} into the error analysis of truncated sum mentioned in \thmref{thm:performance-improvedmc}, and the additive error guarantee can be analogously converted to the relative error guarantee by~\cite{Hu06}. Therefore, our primary goal is to provide the time complexity of this algorithm, which essentially involves assessing the expected running time of sampling an incoming directed spanning tree.

    As shown in \algoref{algo:treemc}, the time of sampling an incoming directed spanning tree is determined by the total number of visits to nodes that are not yet included in the sampled incoming directed spanning tree. In the first iteration of the loop-erased random walk starting from node \(i\), the expected number of visits to \(i\) is equal to \(\subfund_{i,i}\).
    Recall that for Wilson's algorithm, the distribution of sampled path is independent of the node ordering~\cite{Wi96}. In other words, any node can be selected in the initial round while the distribution of sampled path remains unchanged.
    Therefore, the expected running time of sampling an incoming directed spanning tree is equal to \(\bigtrace{\subfund}\).

    According to \lemref{lem:trace-approx}, we can obtain the required sample size \(r\). Combining \lemref{lem:diag-truncate} with the error analysis of truncated sum in \thmref{thm:performance-improvedmc}, the necessary amount \(r'\) and length \(l\) of simulating truncated random walks can be determined. Finally we can represent the time complexity \(O(T)\) of \algoname{TreeMC} as
    \begin{align*}
        T & =r\cdot\bigtrace{\subfund}+r'\cdot l                                                                      \\
          & =\epsilon^{-2}\tau^2\dmax^{2\tau}\log^3n\cdot\bigtrace{\subfund}+\frac{l^3\log n}{2\epsilon^2\pi_s^2n^2},
    \end{align*}
    finishing the proof.
\end{proof}
}

As shown in \thmref{thm:performance-treemc}, the expected time for sampling an incoming directed rooted spanning tree scales as \(\bigtrace{\subfund}\), equivalent to the sum of diagonal entries in \(\subfund\). In contrast, the expected cost of a single naive absorbing walk simulation entails summing all entries of the dense matrix \(\subfund\). Therefore, the efficiency gains of \algoname{TreeMC} are significant relative to naive simulation of absorbing random walks.

Meanwhile, \thmref{thm:performance-treemc} also emphasizes the importance of selecting absorbing node \(s\) in \algoname{TreeMC}. First, the expected running time of \algoname{TreeMC} depends on the mean hitting time from nodes in \(\gr\) to \(s\). Enhanced reachability of \(s\) leads to improved efficiency of \algoname{TreeMC}. Additionally, the accuracy of estimating \(\absorb{s}\) is related to the scaling coefficient \(\pi_s^{-1}\). If \(\pi_s^{-1}\) is smaller, then the theoretical accuracy of \algoname{TreeMC} can be reduced. Therefore, we choose \(s\) as the node with maximum \(\pi_s\), which is both reasonable and efficient for implementation.

\section{Numerical Experiments}\label{sec:experim}

\ifthenelse{\boolean{fullver}}{
      In this section, we present experimental results for various real-world networks to demonstrate the efficiency and accuracy of our algorithms for estimating Kemeny's constant of digraphs. First, we outline the experimental settings, including information of experimental environment, as well as the selection of datasets, baselines, and parameters. Then, we provide the numerical results on real-world networks, which validate the superiority of our algorithms.
}

\subsection{Experimental Settings}\label{subsec:exp-sett}

\textbf{Datasets.}
The data of real-world digraphs utilized in our experiments are sourced from the Koblenz Network Collection~\cite{Ku13} and SNAP~\cite{LeKr14}. To facilitate fair comparison against previous works, we also conduct experiments on several undirected networks considered in~\cite{XuShZhKaZh20,LiHuLe21}.
For those networks that are not originally strongly connected, we perform our experiments on their largest strongly connected components (LSCCs). Relevant information about the LSCC of studied real-world networks is shown in \tabref{tab:real-time}, where networks are listed in ascending order by the number of nodes. The smallest network has 8490 nodes, while the largest one consists of more than 8 million nodes.

\begin{table*}
      \centering
      \fontsize{8.0}{8.8}\selectfont
      \def\app{\algoname{ApprKem}}
      \def\dyn{\algoname{DynMC}}
      \def\ada{\algoname{AblatMC}}
      \def\refi{\algoname{ImprovedMC}}
      \def\acc{\algoname{TreeMC}}
      \def\Undi{\multirow{5}{*}{\rotatebox[origin=c]{90}{Undirected}}}
      \def\Dir{\multirow{15}{*}{\rotatebox[origin=c]{90}{Directed}}}
      \caption{The running time (seconds, \(s\)) of \algoname{ApproxKemeny} (\app), \algoname{DynamicMC} (\dyn), \algoname{AblationMC} (\ada), {\refi} and {\acc} with various \(\epsilon\) on realistic networks.}
      \label{tab:real-time}
      \begin{tabular}{@{}ccrrccccccccc@{}}
            \toprule
            \multirow{3}{*}{Type}
                  &
            \multirow{3}{*}{Network}
                  &
            \multirow{3}{*}{Node}
                  &
            \multirow{3}{*}{Edge}
                  &
            \multicolumn{9}{c}{Running time (s)}                                                                                                                                                          \\
            \cmidrule(l){5-13}

                  &                 &           &             &
            \multirow{2}{*}{\app}
                  &
            \multirow{2}{*}{\dyn}
                  &
            \multirow{2}{*}{\ada}
                  &
            \multicolumn{3}{c}{\refi}
                  &
            \multicolumn{3}{c}{\acc}                                                                                                                                                                      \\
            \cmidrule(l){8-13}

                  &                 &           &             &       &       &       & \(\epsilon=0.3\) & \(\epsilon=0.2\) & \(\epsilon=0.15\) & \(\epsilon=0.3\) & \(\epsilon=0.2\) & \(\epsilon=0.15\) \\
            \midrule
            \Undi & Sister cities   & 10,320    & 17,988      & 0.714 & 1.246 & 0.883 & 0.252            & 0.586            & 1.104             & 0.011            & 0.027            & 0.036             \\
                  & PGP             & 10,680    & 24,316      & 0.767 & 1.230 & 0.792 & 0.222            & 0.508            & 0.927             & 0.008            & 0.016            & 0.027             \\
                  & CAIDA           & 26,475    & 53,381      & 1.274 & 1.507 & 1.408 & 0.268            & 0.610            & 1.107             & 0.013            & 0.025            & 0.042             \\
                  & Skitter         & 1,694,616 & 11,094,209  & --    & 437.0 & 92.78 & 0.973            & 1.534            & 2.102             & 0.829            & 1.795            & 3.000             \\
                  & Orkut           & 3,072,441 & 117,184,899 & --    & 446.3 & 223.2 & 1.854            & 2.873            & 3.930             & 3.022            & 5.849            & 10.03             \\
                  & soc-LiveJournal & 5,189,808 & 48,687,945  & --    & 2843  & 363.2 & 2.351            & 3.610            & 4.921             & 4.802            & 9.920            & 16.82             \\
            \midrule
            \Dir  & Gnutella30      & 8,490     & 31,706      & --    & 1.138 & 0.980 & 0.292            & 0.714            & 1.272             & 0.007            & 0.014            & 0.028             \\
                  & Wikilink-wa     & 22,528    & 135,661     & --    & 1.388 & 1.097 & 0.238            & 0.509            & 1.005             & 0.015            & 0.028            & 0.052             \\
                  & Epinions        & 32,223    & 443,506     & --    & 2.760 & 1.545 & 0.305            & 0.611            & 1.256             & 0.022            & 0.046            & 0.075             \\
                  & EU Inst         & 34,203    & 151,132     & --    & 2.636 & 1.486 & 0.229            & 0.594            & 1.087             & 0.025            & 0.041            & 0.072             \\
                  & Wikilink-la     & 158,427   & 3,486,928   & --    & 4.724 & 2.407 & 0.264            & 0.352            & 0.961             & 0.109            & 0.227            & 0.436             \\
                  & Higgs           & 360,210   & 14,102,583  & --    & 14.33 & 8.100 & 0.479            & 0.669            & 0.894             & 0.315            & 0.650            & 1.076             \\
                  & Youtube         & 509,245   & 4,269,142   & --    & 27.41 & 16.10 & 0.583            & 0.925            & 1.236             & 0.306            & 0.686            & 1.183             \\
                  & Pokec           & 1,304,537 & 29,183,655  & --    & 128.4 & 62.22 & 0.795            & 1.229            & 1.635             & 0.839            & 1.836            & 3.153             \\
                  & Stack Overflow  & 1,642,130 & 32,759,694  & --    & 152.5 & 79.31 & 0.870            & 1.379            & 1.875             & 0.882            & 1.936            & 3.350             \\
                  & Wikilink-fr     & 2,311,584 & 113,754,248 & --    & 187.7 & 99.59 & 0.939            & 1.487            & 2.007             & 1.739            & 3.416            & 5.465             \\
                  & DBpedia         & 3,796,073 & 97,783,747  & --    & 344.9 & 164.7 & 1.299            & 2.053            & 2.809             & 2.925            & 5.744            & 9.781             \\
                  & Wikilink-en     & 8,026,662 & 416,705,115 & --    & --    & 445.4 & 2.778            & 3.952            & 5.360             & 9.386            & 18.99            & 33.23             \\
            \bottomrule
      \end{tabular}
      \undef{\app}
      \undef{\dyn}
      \undef{\ada}
      \undef{\refi}
      \undef{\acc}
      \undef{\Undi}
      \undef{\Dir}
\end{table*}

\textbf{Environment.}
All experiments are conducted on a Linux server with 32-core 2.5GHz CPU and 256GB of RAM. We implement all the algorithms in Julia for a fair comparison. For \algoname{ApproxKemeny}, we leverage the Laplacian Solver from~\cite{KySa16}. Since all the algorithms are pleasingly parallelizable, we force the program to run on 32 threads in every experiment.

\textbf{Baselines and Parameters.}
To showcase the superiority of our proposed algorithms, we implement several existing methods for comparison. First, we implement the dynamic version of the state-of-the-art method \algoname{DynamicMC} presented in~\cite{LiHuLe21} as a baseline.  Moreover, we implement the algorithm \algoname{ApproxKemeny} in~\cite{XuShZhKaZh20}, which is on the basis of the Laplacian solver~\cite{KySa16}. Meanwhile, as our proposed algorithm \algoname{ImprovedMC} incorporates two optimization techniques, we need to ensure that both of these techniques meet our expectations. For this purpose, we implement an ablation method called \algoname{AblationMC}, which solely utilizes the adaptive sampling technique while excluding the subset sampling technique.

For \algoname{DynamicMC}, we set the threshold parameter \(\epsilon_d=0.0005n\), which is the same as~\cite{LiHuLe21}. For \algoname{ApproxKemeny}, the error parameter \(\epsilon\) is set to be \(0.2\), which is shown in~\cite{XuShZhKaZh20} to achieve good performance. For \algoname{AblationMC}, the error parameter \(\epsilon\) is also set to be \(0.2\). It is worth noting that \algoname{ImprovedMC} and \algoname{TreeMC} are dependent on the second largest modulus of eigenvalues \(\lambda\) of \(\matp\), and \algoname{TreeMC} additionally depends on the stationary distribution \(\vecpi\),  the left eigenvector associated with eigenvalue 1 for \(\matp\). Therefore, we use the Implicitly Restarted Arnoldi Method~\cite{LeSoYa98} to calculate \(\lambda\) and \(\vecpi\) in advance for each tested network. Our results demonstrate that even for the largest real-world dataset tested, the pre-computation time is much less than one minute, which is negligible compared with the time for calculating Kemeny's constant. Therefore, we do not take the pre-computation time into account in our experiments.

\subsection{Results on Real-world Networks}\label{subsec:res-real}

\subsubsection{Efficiency}\label{subsubsec:res-real-effic}
We first evaluate the efficiency of our proposed algorithms on real-world networks. The execution time of our proposed algorithms and baselines is reported in \tabref{tab:real-time}. Specifically, we present the results of \algoname{ImprovedMC} and \algoname{TreeMC} for \(\epsilon\in\setof{0.3,0.2,0.15}\).
Note that due to the limit of memory space, \algoname{ApproxKemeny} is infeasible for the large undirected graphs, while \algoname{DynamicMC} is infeasible for the largest digraph due to the running time being more than one hour.

\tabref{tab:real-time} indicates that for every real-world network, the running time of \algoname{ImprovedMC} and \algoname{TreeMC} with \(\epsilon=0.2\) is smaller than that of baselines.
Recall that the theoretical running time of \algoname{ImprovedMC} and \algoname{TreeMC} is proportional to \(\epsilon^{-3}\) and \(\epsilon^{-2}\), respectively. As shown in \tabref{tab:real-time}, the larger constant factor \(\epsilon^{-3}\) leads to the phenomenon that \algoname{ImprovedMC} is slower than \algoname{TreeMC} on relatively small networks. However, for large-scale networks like Wikilink-en, the sublinear time complexity of \algoname{ImprovedMC} becomes dominant, leading to evident speedup compared to \algoname{TreeMC}.

Additionally, we observe that the running time of \algoname{DynamicMC} in our experiments is longer than that reported in~\cite{LiHuLe21}. As mentioned in \secref{subsec:exist-methods}, the high efficiency of \algoname{DynamicMC} in~\cite{LiHuLe21} is largely attributed to GPU-based implementation. In our experiments, we implement all baselines and our proposed approaches by using 32 CPU threads, which also ensures a fair comparison.

The results in \tabref{tab:real-time} also reveal that the ablation method \algoname{AblationMC} outperforms \algoname{DynamicMC} in almost every tested network. This advantage indicates that the optimization by adaptive sampling technique is effective. Meanwhile, the speed-up of \algoname{ImprovedMC} compared to \algoname{AblationMC} is also remarkable, which validates the high efficiency of the subset sampling technique.

\subsubsection{Accuracy}\label{subsubsec:res-real-acc}

We next evaluate the accuracy of our proposed algorithms. According to~\eqref{eq:kem-exact}, we can compute the exact value of Kemeny's constant for small real-world networks. The mean relative error of approximation algorithms compared with exact values is reported in \tabref{tab:realistic-error}.

\begin{table}
      \centering
      \tabcolsep=3pt
      \fontsize{6.5}{7.2}\selectfont
      \def\app{\algoname{ApprKem}}
      \def\dyn{\algoname{DynMC}}
      \def\ada{\algoname{AblatMC}}
      \def\refi{\algoname{ImprovedMC}}
      \def\acc{\algoname{TreeMC}}
      \def\myerr{Mean relative error (\(\times{10}^{-3}\))}
      \caption{The mean relative error (\(\times{10}^{-3}\)) of \algoname{ApproxKemeny} (\app), \algoname{DynamicMC} (\dyn), \algoname{AblationMC} (\ada), {\refi} and {\acc}.}
      \label{tab:realistic-error}
      \begin{tabular}{@{}cccccccccc@{}}
            \toprule
            \multirow{3}{*}{Network}  &
            \multicolumn{9}{c}{\myerr}                                                                                      \\
            \cmidrule(l){2-10}        &
            \multirow{2}{*}{\app}     &
            \multirow{2}{*}{\dyn}     &
            \multirow{2}{*}{\ada}     &
            \multicolumn{3}{c}{\refi} &
            \multicolumn{3}{c}{\acc}                                                                                        \\
            \cmidrule(l){5-10}
                                      &       &       &       & \(0.3\) & \(0.2\) & \(0.15\) & \(0.3\) & \(0.2\) & \(0.15\) \\
            \midrule
            Sister cities             & 0.276 & 2.550 & 1.479 & 2.768   & 1.504   & 0.812    & 0.839   & 0.265   & 0.225    \\
            PGP                       & 0.904 & 3.175 & 2.387 & 2.268   & 0.781   & 0.454    & 0.523   & 0.183   & 0.170    \\
            CAIDA                     & 0.103 & 1.996 & 0.367 & 3.231   & 0.912   & 0.772    & 0.157   & 0.075   & 0.054    \\
            \midrule
            Gnutella30                & --    & 0.685 & 0.266 & 0.506   & 0.252   & 0.155    & 0.453   & 0.179   & 0.081    \\
            Wikilink-wa               & --    & 0.427 & 0.094 & 1.703   & 0.167   & 0.112    & 0.239   & 0.049   & 0.027    \\
            Epinions                  & --    & 4.607 & 1.622 & 2.652   & 1.926   & 0.829    & 0.508   & 0.177   & 0.140    \\
            EU Inst                   & --    & 9.261 & 0.295 & 2.364   & 1.739   & 0.632    & 0.279   & 0.215   & 0.165    \\
            \bottomrule
      \end{tabular}
      \undef{\app}
      \undef{\dyn}
      \undef{\ada}
      \undef{\myerr}
      \undef{\refi}
      \undef{\acc}
\end{table}

\tabref{tab:realistic-error} indicates that for all the evaluated algorithms that exhibit a theoretical error guarantee, their estimated relative error is significantly lower than guaranteed value, including \algoname{ImprovedMC} and \algoname{TreeMC}. For the case of \algoname{ImprovedMC} with \(\epsilon=0.2\), its maximum approximation error is less than 0.2\%, which is considerably small. Furthermore, it is evident that \algoname{TreeMC} with \(\epsilon=0.15\) consistently provides the most precise answer, which can be largely attributed to the alternative formula derived in \thmref{thm:kem-submat}.

Finally, the results in \tabref{tab:realistic-error} also indicate that the mean relative error of ablation method \algoname{AblationMC} is always lower than that of \algoname{DynamicMC}.
This discrepancy in empirical accuracy arises from the different selections of algorithm parameters. For \algoname{DynamicMC}, the simulation amount \(r\) is fixed and the truncated length \(l\) is dynamically determined. The error introduced by dynamically determining \(l\) is biased, and this bias cannot be compensated by a fixed large value of \(r\). In contrast, \algoname{AblationMC} fixes \(l\) and dynamically determines \(r\) based on \lemref{lem:bernstein}, whose incurred error is unbiased.
Recall that \algoname{AblationMC} solely incorporates the adaptive sampling technique, we can confirm that this technique makes significant improvement of efficiency without sacrificing theoretical accuracy. In summary, our proposed algorithms exhibit comparable accuracy with remarkable speed-up, and the algorithm \algoname{TreeMC} is even more accurate than other competitors.

\subsection{Influence of Varying Error Parameter}

In evaluating the efficiency and accuracy of our algorithms, we observe that the error parameter \(\epsilon\) markedly impacts the performance. We now examine in detail how  \(\epsilon\)  affects the efficiency and accuracy. Specifically, we range \(\epsilon\) from 0.15 to 0.4 and provide the running time and mean relative error of each algorithm on several real-world networks. Notably, for \algoname{DynamicMC} the parameter \(\epsilon_d\) relates to threshold rather than error guarantee. We thus fix \(\epsilon_d=0.0005n\) and present the performance of \algoname{DynamicMC} as a baseline.

\subsubsection{Effect on efficiency}

We first assess the impact of varying error parameter on the efficiency of different algorithms. The results on real-world networks are presented in \figref{fig:real-time}.

\begin{figure}[htbp]
      \centering
      \includegraphics[width=\linewidth]{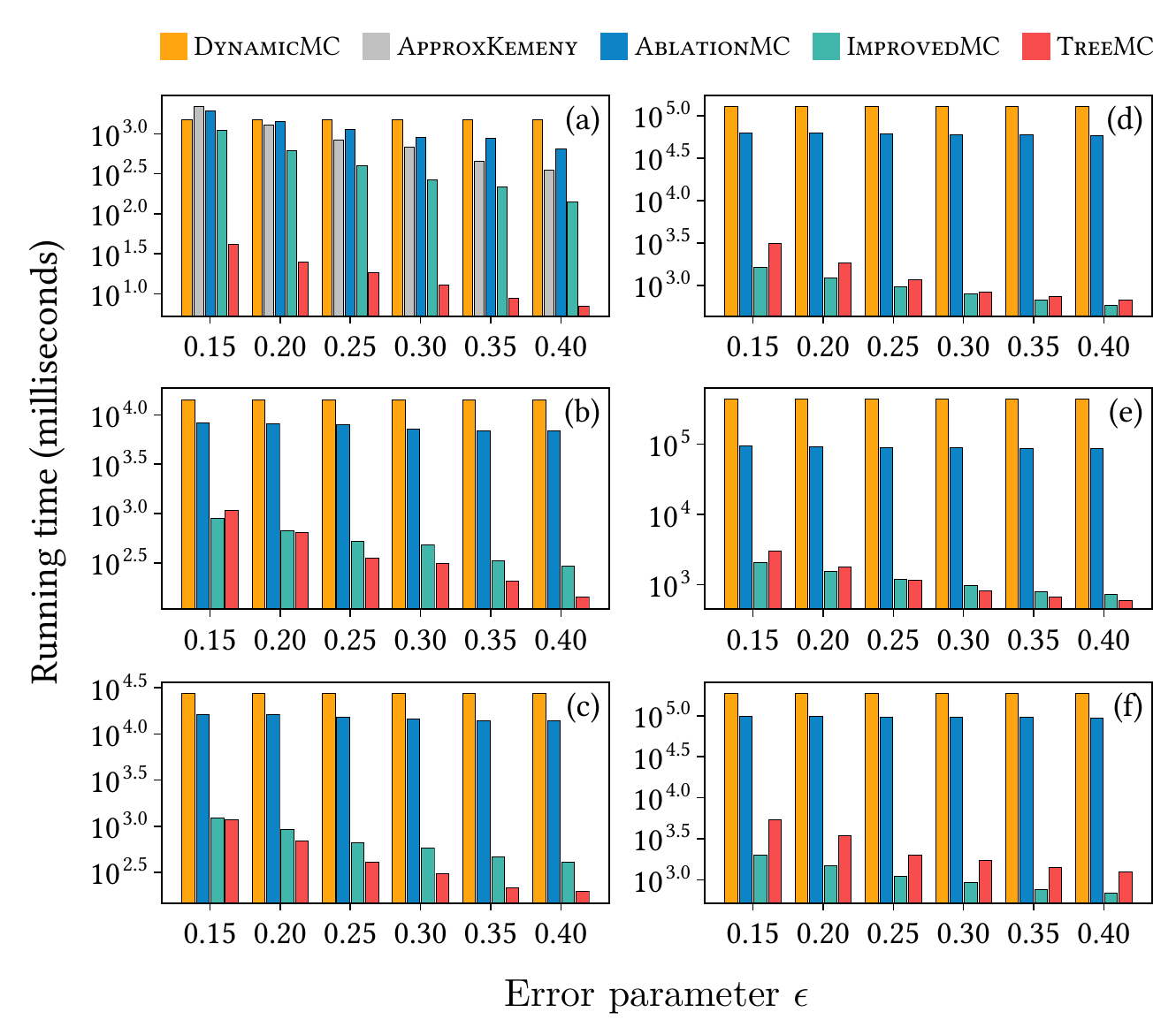}
      \caption{Running time of different approximate algorithms with varying error parameter \(\epsilon\) on real-world networks: CAIDA (a), Higgs (b), Youtube (c), Pokec (d), Skitter (e) and Wikilink-fr (f).}
      \Description{Running time of different approximate algorithms with varying error parameter on real-world networks.}
      \label{fig:real-time}
\end{figure}

As shown in \figref{fig:real-time}, the running time of \algoname{ApproxKemeny} and \algoname{TreeMC} demonstrates comparable growth trends, which is more apparent than competitors. This aligned scaling corroborates the asymptotic complexities proportional to \(\epsilon^{-2}\). For \algoname{AblationMC} and \algoname{ImprovedMC}, their theoretical complexity is proportional to \(\epsilon^{-2}\) and \(\epsilon^{-3}\), respectively. However, such growth patterns are less evident on large networks like Pokec and Skitter. This situation should be attributed to the leverage of adaptive sampling technique, which may terminate the unnecessary simulation in advance. This result further validates the effectiveness of adaptive sampling technique.

\subsubsection{Effect on accuracy}

\begin{figure}[htbp]
      \centering
      \includegraphics[width=\linewidth]{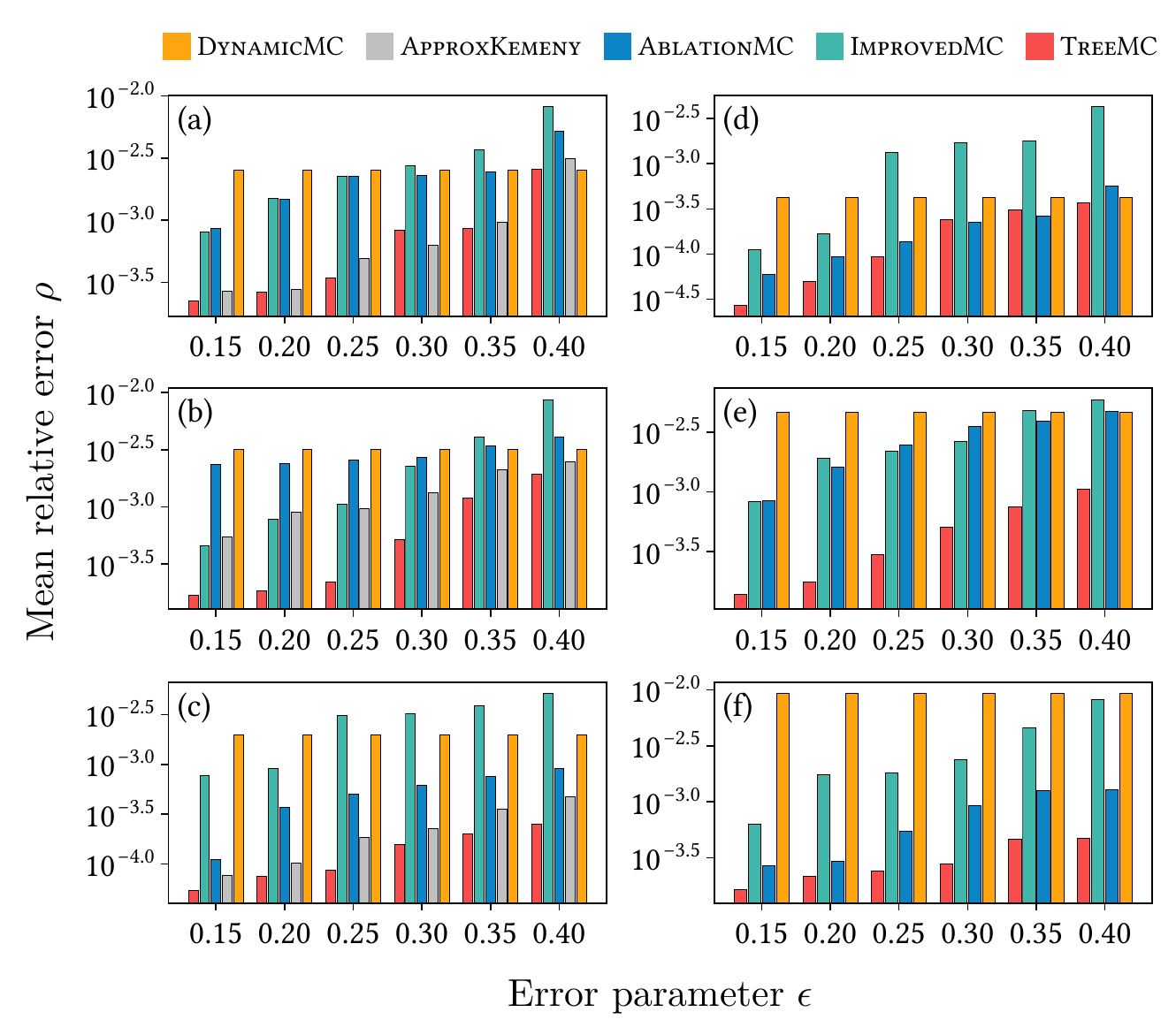}
      \caption{Mean relative error of different approximate algorithms with varying error parameter \(\epsilon\) on real-world networks: Sister cities (a), PGP (b), CAIDA (c), Wikilink-wa (d), Epinions (e) and EU Inst (f).}
      \Description{Mean relative error of different approximate algorithms with varying error parameter on real-world networks.}
      \label{fig:real-error}
\end{figure}

We next analyze the impact of varying \(\epsilon\) on the accuracy across algorithms. The results are presented in \figref{fig:real-error}. As shown in \figref{fig:real-error}, \algoname{TreeMC} consistently yields the highest or the second highest accuracy in estimating Kemeny's constant. Although the error of \algoname{ImprovedMC} with large \(\epsilon\) is not ideal, reducing \(\epsilon\) to 0.2 or 0.15 significantly decreases its error to levels comparable with other methods. Notably, the estimation errors of our proposed algorithms are sensitive as \(\epsilon\) changes, confirming that their accuracy is effectively governed by the error parameter.

\section{Related Work}

\textbf{Other Laplacian solver-based methods.}
In addition to existing methods mentioned in \secref{subsec:exist-methods}, Zhang~\textit{et al.}~\cite{ZhXuZh20} introduced another approximation algorithm \algoname{ApproxHK}, which also leverages the nearly linear-time Laplacian solver~\cite{KySa16}. Unlike \algoname{ApproxKemeny}~\cite{XuShZhKaZh20} that uses Hutchinson's method to avoid direct computation of the pseudoinverse, \algoname{ApproxHK} leverages the Johnson-Lindenstrauss (JL) lemma. However, the total memory requirements of JL lemma and Laplacian solver limit \algoname{ApproxHK}'s scalability compared to \algoname{ApproxKemeny} on large graphs. In contrast, our proposed algorithms mainly use memory to store the network, enabling improved scalability. Meanwhile, the usage of Laplacian solver also restricts \algoname{ApproxHK} to undirected graphs, while our algorithms support both directed and undirected graphs.

\textbf{Other spanning tree-based methods.}
To approximate matrix inverses, researchers have proposed other spanning tree-based approaches. Angriman \textit{et al.}~\cite{AnPrGrMe20} introduced an algorithm for estimating the diagonal elements of the Laplacian pseudoinverse by sampling spanning trees. Specifically, they infused current flows into the sampled spanning trees and used the average value of current flows to estimate resistance distance for undirected graphs, which is a key procedure of their algorithm. However, the definition of resistance distance is limited to undirected graphs, and the theoretical foundation of current flows cannot be directly applied to digraphs. In contrast, our proposed algorithm \algoname{TreeMC} samples rooted spanning trees, which is a distinct approach from that of~\cite{AnPrGrMe20}.

\textbf{Discussion of directed Laplacian solver.}
As discussed in \secref{subsec:exist-methods}, \algoname{ApproxKemeny}~\cite{XuShZhKaZh20} cannot be extended to digraphs due to restrictions of the nearly linear-time Laplacian solver~\cite{KySa16}. Though prior works have proposed nearly linear-time algorithms for solving directed Laplacian systems with errror guarantee~\cite{CoKeKyPePeRaSi18,KyMePr22}, their theoretical efficiency has not been translated to practical implementation. Therefore, it remains infeasible to directly apply these directed Laplacian solvers for efficient estimation of Kemeny's constant on digraphs.

\textbf{Computation of personalized PageRank.}
Apart from the Kemeny constant, other random walk-based quantities have also garnered substantial research attention, such as personalized PageRank (PPR)~\cite{HoGuZhWaWe23,HoChWaWe21,WaWeGaYuDuWe22,ChGuZhYaSk23,WuGaWeZh21}. Since the PPR vector satisfies a recursive relationship, several recent studies~\cite{WaWeGaYuDuWe22,ChGuZhYaSk23} utilized a variety of push-based deterministic approaches for its computation, while others~\cite{HoGuZhWaWe23,WuGaWeZh21} combined push-based approaches with Monte Carlo methods.
Given that the PPR stems from the \(\alpha\)-decaying random walk, where the walker may stop at each visited node with probability \(\alpha\), the expected running time of push-based methods scales in proportion to \(\alpha^{-1}\). In contrast, simple random walks related to Kemeny's constant can be extremely long. Thus, directly adapting existing push-based algorithms is impractical for efficient Kemeny's constant approximation.

\section{Conclusion}

We presented two different Monte Carlo algorithms \algoname{ImprovedMC} and \algoname{TreeMC} for approximating Kemeny's constant effectively and efficiently. \algoname{ImprovedMC} is based on the simulation of truncated random walks, which utilizes an adaptive sampling technique as well as a technique that allows the simulation from a subset of nodes. Due to these optimization techniques, \algoname{ImprovedMC} exhibits sublinear time complexity while retaining theoretical accuracy.
In order to further improve the accuracy of estimating Kemeny's constant, we proposed \algoname{TreeMC} with the help of an alternative formula in terms of the inverse of a submatrix associated with the transition matrix. Extensive numerical results indicate that \algoname{ImprovedMC} is extremely faster than the state-of-the-art method with comparable accuracy, and that \algoname{TreeMC} outperforms the state-of-the-art method in terms of both efficiency and accuracy, but is slightly slower than \algoname{ImprovedMC}.

\section*{Acknowledgements}

This work was supported by the National Natural Science Foundation of China (Nos. 62372112, U20B2051, and 61872093).

\bibliographystyle{ACM-Reference-Format}
\balance

\appendix

\ifthenelse{\not\boolean{fullver}}{
      \section{Proofs of Lemmas and Theorems}



      \subsection{Proof of \lemref{lem:diag-truncate}}
      \begin{proof}
    Since \(\tau\) represents the longest distance between all pairs of nodes, the inequality \(\matp^{\tau}_{i,j}\geq\dmax^{-\tau}\) holds for any nodes \(i,j\in V\). Plugging this into \lemref{lem:dist-deviat}, we obtain
    \begin{align*}
        \abs{\fund_{i,i}-\truncfund_{i,i}}
         & =\abs{\sum_{k=l+1}^{\infty}\matp^k_{i,i}-\pi_i}\leq\sum_{k=l+1}^{\infty}\abs{\matp^k_{i,i}-\pi_i}    \\
         & \leq\sum_{k=l+1}^{\infty}\spar{1-n\dmax^{-\tau}}^{\floor{\nicefrac{k}{\tau}}}                        \\
         & \leq\frac{\tau\spar{1-n\dmax^{-\tau}}^{\floor{\nicefrac{(l+1)}{\tau}}}}{n\dmax^{-\tau}}\leq\epsilon,
    \end{align*}
    where the last inequality is due to \(l\geq\frac{\tau\log\spar{n\epsilon\tau^{-1}\dmax^{-\tau}}}{\log\spar{1-n\dmax^{-\tau}}}+\tau-1\).
\end{proof}


      \subsection{Proof of \lemref{lem:sqrt-sampling-unbiased}}
      \begin{proof}
    For a strongly connected \(n\)-node digraph \(\gr=(V,E)\), we denote the set of \(k\)-combinations of \(V\) as \(V_k\). Since the node subset \(\sqrtset\) is sampled uniformly at random, we can represent the expected value of \(\tilde{S}\) as
    \begin{align*}
        \bigmean{\tilde{S}}
         & =\sum_{\sqrtset\in V_k}\frac{1}{\binom{n}{k}}\frac{n}{k}\sum_{i\in\sqrtset}\meanretnum{i}                    \\
         & =\frac{1}{\binom{n}{k}}\frac{n}{k}\sum_{u\in V}\binom{n-1}{k-1}\meanretnum{u}=\sum_{u\in V}\meanretnum{u}=S,
    \end{align*}
    which completes our proof.
\end{proof}

      \subsection{Proof of \lemref{lem:sqrt-sampling-error}}
      \begin{proof}
    For \(i=1,2,\dots,n\), let \(y_i\) be Bernoulli random variable such that \(\Pr\spar{y_i=1}=p\) and \(\Pr\spar{y_i=0}=1-p\). Here \(y_i=1\) indicates that \(x_i\) is selected and \(y_i=0\) otherwise. Let \(z_i=x_iy_i\) be \(n\) independent random variables with \(0\leq z_i\leq b\). Denote the sum of random variables \(y_i\) as \(y\), and denote the sum of random variables \(z_i\) as \(z\). Namely, \(y=\sum_{i=1}^{n}y_i\) and \(z=\sum_{i=1}^{n}z_i\). It is clear that \(y\) and \(z\) represent the number of selected numbers and their sum. Therefore the expectations of \(y\) and \(z\) can be expressed as \(\mean{y}=np\) and \(\mean{z}=px\). According to \lemref{lem:hoeffding}, we have
    \begin{equation*}
        \Pr\spar{\abs{\tilde{x}-x}\geq n\epsilon}=\Pr\spar{\frac{\abs{z-px}}{p}\geq n\epsilon}\leq2\exp\setof{-\frac{2p^2n^2\epsilon^2}{nb^2}}\leq\frac{1}{n},
    \end{equation*}
    finishing the proof.
\end{proof}

      \subsection{Proof of \thmref{thm:performance-improvedmc}}

      \subsection{Proof of \lemref{lem:tree-sampling-unbiased}}
      \begin{proof}
    In the initial round of the loop-erased random walk that starts from node \(i\), there is only one absorbing node \(s\) and the expected passage times over \(i\) is \(\subfund_{i,i}\). Hence~\eqref{eq:tree-sampling-unbiased} holds true for the initial starting node. For Wilson's algorithm, the distribution of sampled random walk path is independent of the node ordering~\cite{Wi96}. Therefore, every node in \(V\setminus\setof{s}\) can be sampled in the first round, which indicates that~\eqref{eq:tree-sampling-unbiased} holds true for all \(i\in V\setminus\setof{s}\).
\end{proof}

      \subsection{Proof of \lemref{lem:passnum-upperbound}}
      \begin{proof}
    Since \(\passnum{i}{j}\) denotes the passage times of node \(i\) for the \(\myord{j}\) loop-erased random walk, it is easy to verify that \(\passnum{i}{j}\leq\randhit{i}{s}/2\), where \(\randhit{i}{s}\) denotes the hitting time from \(i\) to \(s\) for the \(\myord{j}\) loop-erased random walk. Therefore, we turn to bound \(\randhit{i}{s}\) with high probability, which requires us to provide an upper bound for its expected value \(\hit{i}{s}\).

    Recall that \(\hit{i}{s}\) can be expressed as \(\vece_i^\top\subfund\vecone\), we have
    \begin{equation}\label{eq:hitting-upperbound}
        \begin{split}
            \max_{i\in V\setminus\setof{s}}\hit{i}{s}
             & =\Abs{\subfund\vecone}_\infty=\sum_{k=0}^{\infty}\Abs{\matp^k_{-s}}_\infty                                                        \\
             & \leq\sum_{k=0}^{\infty}\spar{1-\dmax^{-\tau}}^{\floor{\nicefrac{k}{\tau}}}=\frac{\tau}{1-\spar{1-\dmax^{-\tau}}}=\tau \dmax^\tau.
        \end{split}
    \end{equation}
    According to~\cite{AlFi02}, we can finish our proof by providing the upper bound for \(\randhit{i}{s}\) with high probability:
    \begin{align*}
        \Pr\spar{\passnum{i}{j}>t}
         & =\Pr\spar{\randhit{i}{s}>2t}\leq\exp\setof{-\floor{\frac{2t}{\rme\hit{i}{s}}}} \\
         & =\exp\setof{-\floor{\frac{2t}{\rme\tau \dmax^\tau}}}\leq\frac{1}{4n^2},
    \end{align*}
    where the last inequality is due to \(t\geq\rme\tau \dmax^\tau\ceil{\log\spar{4n^2}}/2\).
\end{proof}

      \subsection{Proof of \lemref{lem:trace-approx}}
      \begin{proof}
    \lemref{lem:passnum-upperbound} reveals that \(\passnum{i}{j}\) exhibits an explicit upper bound \(t=\rme\tau \dmax^\tau\ceil{\log\spar{4n^2}}/2\) with a high probability. Plugging this into \lemref{lem:hoeffding}, we obtain
    \begin{align*}
             & \Pr\spar{\abs{\meanpassnum{i}-\mean{\meanpassnum{i}}}\geq\frac{\epsilon}{2}}                                       \\
        =    & \Pr\biggspar{\biggabs{\sum_{j=1}^{r}\passnum{i}{j}-\bigmean{\sum_{j=1}^{r}\passnum{i}{j}}}\geq\frac{r\epsilon}{2}} \\
        \leq & 1-\spar{1-2\exp\setof{-\frac{2r^2\epsilon^2}{4rt^2}}}\spar{1-\frac{1}{4n^2}}                                       \\
        \leq & 1-\spar{1-\frac{1}{4n^2}}^2\leq1-\spar{1-\frac{2}{4n^2}}=\frac{1}{2n^2},
    \end{align*}
    where the second inequality follows from \(r\geq4\epsilon^{-2}t^2\log\spar{2n^2}\). Based on the union bound, it holds that
    \begin{align*}
             & \biggabs{\bigtrace{\subfund}-\sum_{i\in V\setminus\setof{s}}\meanpassnum{i}}               \\
        \leq & \sum_{i\in V\setminus\setof{s}}\abs{\subfund_{i,i}-\meanpassnum{i}}\leq\frac{n\epsilon}{2}
    \end{align*}
    with probability
    \begin{equation*}
        \spar{1-\frac{1}{2n^2}}^n\geq1-n\cdot\frac{1}{2n^2}=1-\frac{1}{2n}.
    \end{equation*}
    This finishes the proof.
\end{proof}

      \subsection{Proof of \thmref{thm:performance-treemc}}
      
}

\end{document}